\documentclass[10pt,twocolumn]{IEEEtran}
\usepackage[margin={1.5cm,1.5cm}]{geometry}
\usepackage{amssymb, amsmath, graphicx, epsfig, booktabs}
\usepackage{amsthm}
\usepackage[usenames,dvipsnames]{pstricks}
\usepackage{epsfig}
\usepackage{pst-grad} 
\usepackage{pst-plot} 
\usepackage{bm}
\usepackage{filecontents}
\usepackage{stfloats}
\usepackage{type1cm}
\usepackage{lettrine}
\usepackage{graphicx}
\usepackage{multicol}
\usepackage{algorithm}
\usepackage{algpseudocode}
\usepackage[font= small]{caption}
\usepackage{booktabs}
\usepackage{dsfont}
\usepackage{mathrsfs}
\usepackage{epstopdf}
\usepackage{subfigure}
\usepackage[labelfont=bf]{caption}

\let\expect\relax
\let\real\relax
\let\prob\relax
\theoremstyle{definition}
\newtheorem{theorem}{\bf Theorem}
\newtheorem{proposition}[theorem]{\bf Proposition}
\newtheorem{lemma}[theorem]{\bf Lemma}

\newtheorem{definition}{\bf Definition}

\newtheorem{remark}{\normalfont \textit{Remark}}

\makeatletter
\newcommand{\pushright}[1]{\ifmeasuring@#1\else\omit\hfill$\displaystyle#1$\fi\ignorespaces}
\newcommand{\pushleft}[1]{\ifmeasuring@#1\else\omit$\displaystyle#1$\hfill\fi\ignorespaces}
\makeatother

\renewcommand{\d}[1]{\ensuremath{\operatorname{d}\!{#1}}}

\renewcommand{\d}[1]{\ensuremath{\operatorname{d}\!{#1}}}

\allowdisplaybreaks[4]

\usepackage{hyperref}
\hypersetup{
	colorlinks=true,
	linkcolor=[RGB]{110,25,25},
	citecolor=[RGB]{25,25,110},
	filecolor=magenta,      
	urlcolor=cyan,
}

\let\triangleq\relax
\newcommand{\triangleq}{\stackrel{\text{def}}{=}}
\newcommand{\real}{{\rm I\!R}}
\newcommand{\expect}{{\rm I\!E}}
\newcommand{\prob}{{\rm I\!P}}

\usepackage{amssymb}

\renewcommand{\d}[1]{\ensuremath{\operatorname{d}\!{#1}}}
\usepackage{mathrsfs}

\begin{document}
\title{On Online Energy Harvesting in Multiple Access Communication Systems}
\author{{Masoud~Badiei~Khuzani, \textit{Student Member}, \textit{IEEE}, and Patrick~Mitran, \textit{Senior Member}, \textit{IEEE}}\\
Department of Electrical and Computer Engineering \\
The University of Waterloo, Ontario, Canada}

\maketitle
\thispagestyle{empty}
{\let\thefootnote\relax\footnotetext{
This work was presented in part at the 2012 IEEE International Symposium on Information Theory \cite{Mitran}.

Email: \textsc{mbadieik@uwaterloo.ca, pmitran@uwaterloo.ca.}}
\setcounter{footnote}{0}

\begin{abstract}
We investigate performance limits of a multiple access communication system with energy harvesting nodes where the utility function is taken to be the long-term average sum-throughput. We assume a causal structure for energy arrivals and study the problem in the continuous time regime. For this setting, we first characterize a storage model that captures the dynamics of a battery with energy harvesting and variable transmission power. Using this model, we next establish an upper bound on the throughput problem as a function of battery capacity. We also formulate a non-linear optimization problem to determine optimal achievable power policies for transmitters. Applying a calculus of variation technique, we then derive Euler-Lagrange equations as necessary conditions for optimum power policies in terms of a system of coupled partial integro-differential equations (PIDEs). Based on a Gauss-Seidel algorithm, we devise an iterative algorithm to solve these equations. We also propose a fixed-point algorithm for the symmetric multiple access setting in which the statistical descriptions of energy harvesters are identical. Along with the analysis and to support our iterative algorithms, comprehensive numerical results are also obtained.
\end{abstract}


 \setcounter{equation}{0}

\section{Introduction}
The direct impact of energy on communication cost and lifetime has spurred significant efforts to manage and optimize energy consumption. In this respect, current and future state of the art technology has focused on harvesting energy from the environment. It is thus of paramount importance to design suitable adaptive power transmission policies for these technologies. In particular, the formulation of power policies in energy harvesting systems depends intricately on many factors, including energy arrival distribution, battery capacity, quality of service, \textit{etc}. Moreover, most renewable energy resources have unpredictable behaviour that makes the design process of optimal power policies difficult. Solar panels, for instance, can only scavenge sunlight during the daytime and even then, this is a function of weather conditions. Another example is thermoelectric generators where energy is absorbed based on random temperature gradients between two metal junctions. Regarding these examples, a key objective of recent studies is to engineer optimal transmission power polices. These studies, depending on causal or non-causal knowledge of future energy arrivals, fall within two major categories: offline or deterministic (for non-causal), and online or stochastic (for causal) energy harvesting systems.

In the offline regime and in terms of throughput maximization, optimal power allocation for different communication topologies has been well studied. For instance,  \cite{Yang} studies the multiple access channel (MAC), \cite{Ulukus} studies the broadcast channel, and the interference channel is studied in \cite{Yener}. In addition, the issue of maximizing throughput in a fading channel has been treated in \cite{Ozel}. There, a directional water-filling algorithm is proposed. In \cite{GG1}, a continuous time energy harvesting system with constant energy leakage rate due to battery imperfections is considered. Another interesting problem has been studied in \cite{Yang-Jing} where an offline energy harvesting problem subject to minimizing the transmission completion time is analyzed. Specifically, a continuous-time policy to minimize the delivery time of data packets is formulated.
Among more recent results in the offline setting is \cite{Gurakan} where energy cooperation in a two-hop communication system is considered.

 As an overview of prior works in the online regime, we refer the reader to \cite{Ozel}, \cite{capacityO}, \cite{Sharma}, and \cite{Ozel-Ulukus}. In \cite{Ozel} an algorithm in the offline problem of throughput maximization by a deadline was heuristically applied to the online counterpart. The authors have also considered a dynamic programming solution for online policies. However, the curse of dimensionality in the backward induction renders the computational cost of this approach very expensive. In \cite{capacityO}, the capacity of the additive white Gaussian noise channel (AWGN) under discrete-time energy arrivals and infinite battery capacity is characterized. Additionally, two achievable schemes based on save-and-transmit and best-effort-transmit are studied there. In \cite{Sharma}, queuing aspects of the online energy harvesting problem with infinite battery and buffer capacity have been considered. The authors have also suggested a greedy policy that in the low signal to noise ratio (SNR) regime is throughput optimal and attains minimum delay. A more relevant study related to the work presented here is \cite{Koksal}. Therein, Srivastava and Koksal have investigated an optimization problem where the objective is to maximize a utility function subject to causality and battery constraints. More interestingly, they addressed a trade-off between achieving the optimum utility and keeping the discharge rate low.

In this paper, we consider the online setting with continuous time policies in which the energy release rates are regulated dynamically based on the remaining charge of the battery at each moment. This architecture naturally requires a different mathematical framework in terms of modelling and analysis. Particularly, the main tool here for modelling the interaction between battery, energy arrivals, and energy consumption is a stochastic process known as a \textit{compound Poisson dam} model. This model was pioneered by Moran in 1954 \cite{Moran} and studied further by Gaver-Miller \cite{Gaver} and Harrison-Resnick \cite{Harrison}. In connection with this model, we derive an upper bound on the total sum-throughput of an online energy harvesting system. Also in terms of achievability, we construct an optimization problem to maximize the sum-throughput subject to an ergodicity constraint. This maximization problem turns out to be non-linear and analytically cumbersome. Relying on a calculus of variations approach, we subsequently find a system of simultaneous PIDEs as necessary conditions for an optimal power policy. We then propose a Gauss-Seidel method (see \cite{Ber99}) to solve these equations efficiently. In the symmetric case, when the statistical description of all the energy harvesters are identical, we obtain an alternative algorithm using a fixed point iteration method. Moreover, in the case of the point-to-point channel setting, the necessary condition further reduces to a non-linear, autonomous ordinary differential equation (ODE) that can be solved directly, using conventional numerical methods \cite{Mitran}.

The rest of the paper is organized as follows. In section II, we review some background, definitions, and notation. Section III deals with necessary and sufficient conditions for ergodicity of the storage process. In Section IV, we derive an upper bound as well as the achievability results for both finite and infinite storage cases, including two algorithms for the achievability part. These algorithms are then used to compute the numerical results in Section V. Lastly, in Section VI, we summarize our main findings and outline possible future directions.
\section{Preliminaries}
\label{section:Preliminaries}
\subsection{Communication model}
We consider $M$ multiple access transmission nodes that wish to transmit their data over a shared communication channel. Furthermore, each transmission node has an energy harvesting module and a battery to capture and store arriving energy packets. Throughout the paper, we denote the instantaneous transmission power at time $t$ from the $k^{th}$ node ($k=1,2,\cdots,M$) by $P_{k}(t)$. Also, to quantify the corresponding transmission rate of the nodes, we consider Shannon's rate function, $r(x)={1\over 2}\log_{2} \big(1+({x/N_{0}})\big)$, where $N_{0}$ denotes the noise power spectral density. In particular, Shannon's rate function carries the following properties and, unless stated otherwise, only these properties will be used in Section IV:
\begin{itemize}
\item[(A.1)] \textit{Positivity}: $r(x)> 0 \ \text{for all}\ x>0 \ \text{and}\ r(0)=0$.

\item[(A.2)] \textit{Differentiability}: $r(x)$ is three times continuously differentiable on $x\geq 0$.

\item[(A.3)] \textit{Monotone increasing}: $r'(x)> 0$ for all $x\geq 0$.

\item[(A.4)] \textit{Concavity}: $r''(x)< 0$ for all $x\geq 0$.
\end{itemize}
Letting $R_{k}$ denote the long-term average rate of the $k^{th}$ user, we then have the rate-region described by
\begin{align}
\label{MAC_region}
\sum_{k\in \mathcal{S}}R_{k}\leq  \lim_{T\rightarrow \infty} {1\over T}\int_{0}^{T}r\left(\sum_{k\in \mathcal{S}}P_{k}(s)\right)\ \text{d}s,
\end{align}
where the inequality holds for all subsets $\mathcal{S} \subseteq \{1,2,\cdots, M \}$, and the resulting region is a polytope called polymatroid. In this study, we restrict ourselves to the dominant face of this polymatroid (called permutahedron) that represents the total sum-throughput (or sum-rate) of the channel. Then, the sum-throughput is
\begin{align}
\label{MAC}
\sum_{k=1}^{M}R_{k}= \lim_{T\rightarrow \infty} {1\over T}\int_{0}^{T}r\left(\sum_{k=1}^{M}P_{k}(s)\right)\d s.
\end{align}
\subsection{Energy harvesting and storage model}
In our energy harvesting model, we allow the transmission nodes to use different techniques for harvesting exogenous energy. For example, while one node may collect solar energy, another node can use a thermoelectric generator. This mechanism is especially important for sensor networks where distributed terminals may measure miscellaneous targets that also feed sensors with energy (\textit{e.g.} see \cite{Sensor}). Mathematically, we assume that for each individual node $k\in \{1,2,\cdots, M \}$, energy is replenished into the corresponding battery according to specific energy arrivals $E_{k}^{0},E_{k}^{1},\cdots$, where the superscript denotes the order of arrivals. Furthermore, the energy arrivals for node $k$ are independent, identically distributed (\textit{i.i.d.}) according to $\prob\{E_{k} \leq x \}=B_{k}(x)$  which occur at random arrival times denoted by $T_{k}^{0},T_{k}^{1},\cdots$. The interarrival times $\Delta T_{k}^{n}=T_{k}^{n+1}-T_{k}^{n}$ are also assumed to be \textit{i.i.d.} and exponentially distributed. Therefore, the attributed point process, $N_{k}(t)\triangleq\sum_{n\in \mathbb{N}}\mathbf{1}_{\{T_{k}^{n}<t\}}$, is a homogeneous Poisson point process with intensity denoted by $\lambda_{k}$. Consequently, the total energy flow $E_{k}^{\mathrm{In}}(0,t]$ into node $k$ and up to time $t$ is a compound Poisson process,
\begin{align}
\label{energy input}
E_{k}^{\mathrm{In}}(0,t]\triangleq \sum_{i=0}^{N(t)}E_{k}^{i}.
\end{align}
To characterize the storage model, we also need to determine the output process at each transmitter. To do so, let $X_{k}(t)$ denote the energy stored in the $k$-th battery as a function of time. Then, the total energy expenditure until time $t$ is
\begin{subequations}
\begin{align}
E_{k}^{\mathrm{Out}}(t)&\triangleq \int_{0}^{t} P_{k}(s)\d s,\\ \label{energy output}
&=\int_{0}^{t} p_{k}(X_{k}(s))\d s,
\end{align}
\end{subequations}
where $p_{k}(\cdot)$ represents the transmission power policy of the $k$-th transmitter, modulated by the available energy in the battery.
Now, the storage equation in terms of the energy arrivals in Eq. \eqref{energy input} and the drift process in Eq. \eqref{energy output} is
\begin{align}
\label{storage model}
&X_{k}(t) = X_{k}(0)+ E_{k}^{\mathrm{In}}(0,t]-\int_{0}^{t} p_{k}(X_{k}(s))\d s,
\end{align}
where $X_{k}(0)$ is the initial battery reserve at time $t=0$, and here the battery is assumed to have infinite capacity $(X_{k}(t)\in [0,\infty))$. In the case that the $k$-th battery has a finite storage capacity, say $L_{k}$, then $X_{k}\in [0,L_{k}]$, and we can similarly characterize the following dynamics,
\small\begin{align}
\label{storage model1}
X_{k}(t) = X_{k}(0)+ E_{k}^{\mathrm{In}}(0,t]-\int_{0}^{t} p_{k}(X_{k}(s))\d s-Z_{k}(t),
\end{align}\normalsize
where $Z_{k}(t)$ is $\real^{+}$ valued process that is null at zero $(Z_{k}(0)=0)$, non-decreasing, continuous almost everywhere, and such that $\int_{\real^{+}}(L_{k}-X_{k}(s))\d Z_{k}(s)=0$. This process, known as reflection process \cite{Mazumdar}, ensures that for any energy arrival, the storage process remains inside the  boundary, \textit{i.e.}, $X_{k}(t)\in [0,L_{k}]$.
%
%

It is also interesting to note that the application of the structures in Eqs. \eqref{storage model} and \eqref{storage model1} are not limited to the current problem. In fact, this formulation has wide applicability in other fields of studies. Examples include workload modulated queues \cite{Workload-Modulated}, water reservoir dam analysis \cite{Asmussen}, food contaminants exposure in bioscience \cite{Food-contaminant}, \textit{etc}. In this paper, the ergodicity results of \cite{Asmussen} will be used and are summarized in section III.

\vspace{1mm}
\textit{Notation.} In the rest of the paper and for conciseness, we adopt several shorthand notations. In particular, $[M]$ stands for $\{1,2,\cdots , M\}$. For $M>1$, we define the rectangular domain $\mathcal{A}$ as
\begin{align}
\nonumber
\mathcal{A} \triangleq [0,L_{1}]\times [0,L_{2}]\times \cdots \times [0,L_{M}].
\end{align}
Related to this, we also define the $M$ dimensional integral by
\begin{align}
\nonumber
\int \limits_{0}^{L_{1}}\int \limits_{0}^{L_{2}}\cdots \int \limits_{0}^{L_{M}}(\cdot)\ \d x_{1}\d x_{2}\cdots \d x_{M},
\end{align}
which is represented by $\int_{\mathcal{A}} (\cdot)\ d\underline{x}$. For all subsets $\mathcal{S}\subseteq [M]$, we use $\mathcal{A}(\mathcal{S})$ to denote the projection of $\mathcal{A}$ onto the coordinates indexed by $\mathcal{S}$, \textit{i.e.},
\begin{align}
\nonumber
\mathcal{A}(\{1,3\})= [0,L_{1}]\times [0,L_{3}].
\end{align}
Then, $\int_{\mathcal{A}(\mathcal{S})} (\cdot)\ d\underline{x}$ denotes integration over a subset of $\real^{\vert \mathcal{S}\vert}$. $\mathcal{A}_{j}$ is also a shorthand for
\begin{align}
\nonumber
\mathcal{A}_{j} \triangleq [0,L_{1}]\times\cdots [0,L_{j-1}]\times [0,L_{j+1}] \cdots \times[0,L_{M}].
\end{align}

\section{Ergodic Theory of Storage Process}
We here summarize necessary and sufficient conditions for ergodicity of the storage process in Eq. \eqref{storage model}. Before stating the definitions regarding ergodic behaviour, we first put some mild constraints on the transmission policies. Particularly, for all $k=1,2,\cdots, M$,
\begin{enumerate}
\item $\forall L_{k}>0, 0< x_{k}\leq L_{k}\Rightarrow p_{k}(x_{k})>0 \ \text{and}\ p_{k}(0)=0, $
\item $\forall L_{k}>0, \sup \limits_{0<x_{k}\leq L_{k}} p_{k}(x_{k})<\infty.$
\end{enumerate}

The first condition indicates that as long as there is energy in the battery, transmission continues (otherwise, the battery would have a minimum energy reserve that can not be consumed). The second condition does not permit the energy in the battery to be consumed instantly. Regarding these constraints, we say a policy is admissible iff it fulfills these two conditions.

\begin{definition}(\textit{Hitting Time})
The hitting time, $\tau(x)$, is defined as the first time that the energy level in the battery reaches the value of $x$. More specifically,
\begin{align}
\nonumber
\tau(x)\triangleq \inf \{t\geq 0: X(t)=x \}.
\end{align}
\end{definition}
\begin{definition}
(\textit{Transient and Recurrent Process} {\cite[p. 290]{Asmussen}}) The storage process is said to be \textit{transient}, if and only if for all initial energy levels $x(0)$ in the battery, we have $\prob(X_{t}\rightarrow \infty)=1$. Alternatively, the storage process is said to be \textit{recurrent} if and only if $\prob[\tau(x)<\infty\vert x(0)]=1$, $\forall x>0,\ x(0)\geq 0$. In the case of a recurrent storage process, it is said to be \textit{positive recurrent} if it further satisfies $\expect[\tau(x)\vert x(0)=x]< \infty$ for one $x>0$ and therefore for all $x>0$ (irreducibility). Similarly, the recurrent storage process is \textit{null recurrent} if $\expect[\tau(x)\vert x(0)=x]= \infty$ for one $x>0$ and therefore for all $x>0$.
\end{definition}

 One motivation for surveying ergodic conditions is to rule out policies that result in transient and null recurrent battery behaviours. For example in the transient case $X(t)\rightarrow \infty$ a.s. which is unrealistic. Also, in the null recurrent case $\lim_{t\rightarrow \infty}\prob\{X(t)\leq u \vert x(0)=x\}= 0, \forall x,u\geq 0$ which implies an unbounded energy reserve in the battery.

\begin{theorem}(\textit{Ergodicity Condition} \cite[Thm. 3.6]{Asmussen})\ The storage process $\{X_{k}(t)\}_{t\geq 0}$ is positive recurrent if and only if there exist a probability measure $\pi_{k}$ that is absolutely continuous on $(0,\infty)$ and which may possess an atom at zero, $\pi_{k}^{0}=\pi_{k}(\{0\})$, \textit{i.e.},
\begin{align}
\label{dens-1}
\pi_{k}(x_{k})=\pi_{k}^{0}+\int_{0^{+}}^{x_{k}}f_{k}(v_{k})\d v_{k},
\end{align}
and such that
\begin{align}
\nonumber
f_{k}(x_{k})=&{\lambda_{k} \over p_{k}(x_{k})}\Big(\pi_{k}^{0}(1-B_{k}(x_{k})) \\ \label{density}
            & +\int_{0^{+}}^{x_{k}}(1-B_{k}(x_{k}-v_{k}))f_{k}(v_{k})\d v_{k} \Big).
\end{align}
Furthermore, $\pi_{k}$ is the unique stationary distribution of the process $X_{k}(t)$.
\hfill $\square$
\end{theorem}
\begin{remark}
The elegant proof of Assmussen for the converse part of Theorem 1 is based on an embedded Markov chain $\{X_{k}(n)\}$ at marked arrival times. In particular, for recurrent embedded chains, it is shown that any storage interval $(x_{k}^{0},x_{k}^{1}),0<x_{k}^{0}<x_{k}^{1}$ is recurrent in the sense of Harris. An alternative proof of the converse part of Theorem 1 adopts the additional condition $\int_{0}^{x_{k}}{(1/p_{k}(u))}\d u<\infty,\forall x_{k}>0$. Due to this extra condition, the required time to reach the zero state in the absence of new arrivals from any energy level in the battery must be finite. For this constraint, it can also be shown that $x_{k}=0$ is a regenerative recurrent point for the process and therefore, due to the additional constraint, the probability measure has a strict atom $\pi_{k}^{0}>0$ at zero.
\end{remark}
\begin{remark}
As discussed in \cite[p. 297]{Asmussen}, in the finite energy case ($L_{k}< \infty$), the storage process is always positive recurrent and the probability measure is likewise governed by Eqs. \eqref{dens-1} and \eqref{density}.
\end{remark}
\begin{remark}
We note that the atom of the probability measure $\pi_{k}(x_{k})$ corresponds to an absorbing state of the process $X_{k}(t)$ in the sense that upon $X_{k}(t)$ entering state $x_{k}=0$, the process remains there until an energy arrival occurs (at which point the process transits to another state). Based on this and the first constraint on admissible power policies (in particular $p_{k}(L_{k})>0$), there is no atom at $x_{k}=L_{k}$ in the finite case since it has a strictly negative drift in Eq. \eqref{storage model1} that shifts the process to the inner region of the state-space instantaneously, i.e., $x_{k}<L_{k}$. Therefore, the battery never idles with $x_{k}=L_{k}$ (reflecting boundary).
\end{remark}

An interpretation for the forward equation in Eq. \eqref{density} can be provided in terms of level crossing theory. In particular,
\begin{align}
\label{reflecct}
&f_{k}(x_{k})p_{k}(x_{k})= \\ \nonumber
&{\lambda_{k}}\left \{\pi_{k}^{0}\big(1-B_{k}(x_{k})\big)+\int_{0^{+}}^{x_{k}}\hspace{-1mm}\big(1-B_{k}(x_{k}-v_{k})\big)f_{k}(v_{k}) \d v_{k} \right \},
\end{align}
is the equilibrium condition between the rate of down crossing at level $x_{k}$ (the \textit{l.h.s} of Eq. \eqref{reflecct}) and up crossing at level $x_{k}$ (the \textit{r.h.s} of Eq. \eqref{reflecct}). We can also view \eqref{density} as a Volterra integral equation of the second kind with the kernel function $K(x_{k},v_{k})=1-B_{k}(x_{k}-v_{k})$, and it can thus be solved numerically (see \cite{PLinz}).

In this paper, we consider the energy arrivals $\{E_{k}^{i}\}_{i=0}^{\infty}, k=1,2,\cdots, M$, to be exponentially distributed with parameter $\zeta_{k}$. Therefore, we have
\begin{align}
\nonumber
K(x_{k},v_{k})=\exp(-\zeta_{k} (x_{k}-v_{k})),
\end{align}
that simplifies \eqref{density} to
\begin{align}
\label{Forward equation summarized form}
f_{k}(x_{k})={\lambda_{k}\exp (-\zeta_{k} x_{k}) \over p_{k}(x_{k})}G_{k}(x_{k}),
\end{align}
where
\begin{align}
\label{F}
G_{k}(x_{k})\triangleq \left(\pi_{0}+\int_{0^{+}}^{x_{k}}\exp (\zeta_{k} v_{k}) f_{k}(v_{k}) \ \d v_{k} \right).
\end{align}
\begin{remark}
\label{Remark:1}
The storage models in Eqs. \eqref{storage model} and \eqref{storage model1} are memoryless, in the sense that at each time instant $t$, the power policy $p_{k}$ only depends on the available charge $X_{k}(t)$ in the battery and not the entire sample path $\{X_{k}(s);s\leq t\}$. As an extension, we can also define a storage model with memory and infinite battery capacity as follows
\begin{align}
X_{k}(t)=X_{k}(0)+E_{k}^{\mathrm{In}}(0,t]-\int_{0}^{t}p_{k}(X_{k}(u);u\leq s)\d s.
\end{align}
The extension of the storage model with memory and finite battery capacity follows similarly. However, when the arrival process is Poisson, it can be shown that $X_{k}(t)$ is a sufficient statistic for an optimal power policy for both infinite and finite battery cases (see Appendix A). In this regard, knowledge of the entire path $\{X_{k}(s);s\leq t\}$ as an argument of $p_{k}(\cdot)$ is excessive.
\end{remark}

\section{Bounds on Total Average Throughput}
\label{Section:Bounds on Total Average Throughput}
Our objective now is to derive an upper bound on the average throughput as well as achievable policies with good performance. In connection with our system model, we will analyze a MAC with 1) finite, and 2) infinite storage batteries.

In particular, in the finite storage case, a good power policy must manage overflow in the battery as regular overflow causes energy waste and potentially decreases the sum throughput. To reduce overflow, the power policy must result in a large transmission power when the battery charge is large as otherwise overflow is likely to occur upon a new arrival. However, transmitting with \textit{too} large a transmission power when the battery happens to have large charge is also undesirable due to the concavity of the rate function. In other words, there is a tension between overflow and the rate at which the large battery charge is consumed to reduce overflow likelihood.

 To further clarify the latter point, consider an energy harvesting system with a single node ($M=1$) in which energy $E$ is replenished into a battery exactly every $T$ units of time. In addition, assume that the transmitter sends data by using a constant transmission power $P=E/(\alpha T), \alpha>0$.  Two cases can now be examined:

(\textit{i}) $\alpha>1$: In this case, the transmitter fails to consume the entire battery charge before the next arrival, and thus overflow occurs regularly. We then have
\begin{align}
T\times r\left({E\over\alpha T}\right)\leq T\times r\left({E\over T}\right).
\end{align}

(\textit{ii}) $\alpha<1$: In this case, the transmitter depletes its available battery charge within $\alpha T< T$ of each arrival. From the concavity of the rate function, we have the following inequality
\begin{align}
\alpha T\times r\left({E\over\alpha T}\right)\leq T\times r\left({E\over T}\right).
\end{align}

 Here, the tension between (\textit{i}) and (\textit{ii}) is resolved by the optimal choice of $\alpha=1$, \textit{i.e.}, $P=E/T$.
\subsection{An Upper Bound}
\subsubsection{Finite Storage Battery}
In this case  $L_{k}< \infty, \forall k\in [M]$. Then from \eqref{MAC} and due to ergodicity of the storage processes  $\{X_{k}(t)_{t\geq 0}\}_{k=1}^{M}$ in the finite battery case (ref. Remark 3), we have almost surely
\begin{align}
\sum_{k=1}^{M} R_{k}\stackrel{\text{a.s.}}{=}\expect\left[r\left(\sum_{k=1}^{M}p_{k}(X_{k})\right)\right],
\end{align}
where the expectation is with respect to the stationary distribution in Theorem 1. In addition, from the concavity property of the rate function and Jensen's inequality,
\begin{align}
\label{Converse_Jensen}
\expect\left[r\left(\sum_{k=1}^{M}p_{k}(X_{k})\right)\right] \leq  r\left(\sum_{k=1}^{M}\expect[p_{k}(X_{k})]\right).
\end{align}
It thus remains to bound the mean transmission power $\expect[p_{k}(x_{k})]$. This can be accomplished by integrating by parts as follows
\begin{align}
\expect[p_{k}(X_{k})] &=\pi_{k}^{0}p_{k}(0)+\int_{0^{+}}^{L_{k}}p_{k}(x_{k})f_{k}(x_{k})\d x_{k} \\
&\stackrel {\rm{(a)}}{=}\int_{0^{+}}^{L_{k}}p_{k}(x_{k})f_{k}(x_{k})\d x_{k}\\ \label{from_G}
&\stackrel{\rm{(b)}}{=}\lambda_{k}\int_{0^{+}}^{L_{k}}\exp (-\zeta_{k}x_{k})G_{k}(x_{k})\d x_{k}\\ \label{expectation}
&= -{\lambda_{k}\over \zeta_{k}}\exp (-\zeta_{k}x_{k})G_{k}(x_{k})\Big \vert_{0^{+}}^{L_{k}} \\ \nonumber
&\hspace{4mm}+ {\lambda_{k}\over \zeta_{k}}\int_{0^{+}}^{L_{k}}\exp (-\zeta_{k}x_{k})G'_{k}(x_{k})\d x_{k},
\end{align}
where $\rm{(a)}$ comes from the first constraint on the admissible power policies and $\rm{(b)}$ follows from \eqref{Forward equation summarized form}. Now from \eqref{F},
\begin{align}
\label{F_der}
G'_{k}(x_{k})=f_{k}(x_{k})\exp(\zeta_{k}x_{k}).
\end{align}
Also we note that $G_{k}(0^{+})=\pi_{k}^{0}$ and
\begin{align}
e^{-\zeta_{k}L_{k}}G_{k}(L_{k})&=e^{-\zeta_{k}L_{k}}\big(\pi_{k}^{0}+\int_{0^{+}}^{L_{k}}\hspace{-1mm}e^{\zeta_{k}x_{k}}f_{k}(x_{k})\d x_{k}\big)\\ \label{Eq:explaininequality}
&\stackrel{\rm{(c)}}{\geq} e^{-\zeta_{k}L_{k}}\big(\pi_{k}^{0}+\int_{0^{+}}^{L_{k}}f_{k}(x_{k})\d x_{k}\big)\\
 \label{integ-bypart-second}
&=e^{-\zeta_{k}L_{k}},
\end{align}
where inequality \rm{(c)} is due to the fact that $\exp(\zeta_{k}x_{k})\geq 1$ for all $x_{k}\geq 0$ since $\zeta_{k}>0$. Substituting \eqref{integ-bypart-second} and \eqref{F_der} in Eq. \eqref{expectation} yields
\begin{align}
\nonumber
\expect[p_{k}(X_{k})]&= {\lambda_{k}\over \zeta_{k}}\big(G_{k}(0^{+})-e^{\zeta_{k}L_{k}}G_{k}(L_{k})\big) \\
&\hspace{4mm}+ {\lambda_{k}\over \zeta_{k}}\int_{0^{+}}^{L_{k}}f_{k}(x_{k})\d x_{k},                  \\
&\leq {\lambda_{k}\over \zeta_{k}}\big(\pi_{k}^{0}-e^{-\zeta_{k}L_{k}}+\int_{0^{+}}^{L_{k}}f_{k}(x_{k})\d x_{k}\big)\\ \label{lambda_zeta} &={\lambda_{k}\over \zeta_{k}}(1-\exp(-\zeta_{k}L_{k})),
\end{align}
In the last step, we now use \eqref{lambda_zeta} and the non-decreasing property of the rate function to characterize an upper bound for all $L_{k}<\infty$ as follows
\begin{align}
\label{upper bound}
\sum_{k=1}^{M}R_{k} \leq r\Big(\sum_{k=1}^{M}{\lambda_{k}\over \zeta_{k}}(1-e^{-\zeta_{k}L_{k}})\Big)\triangleq R_{\text{upper}}.
\end{align}
\subsubsection{Infinite Storage Battery} We now take $L_{k}=\infty$. In this case, similar to \eqref{from_G} we can directly compute,
\begin{align}
&\expect[p_{k}(x_{k})]=\lambda_{k} \int_{0^{+}}^{\infty}e^{-\zeta_{k}x_{k}}G_{k}(x_{k})\d x_{k}\\ &=\lambda_{k} \int_{0^{+}}^{\infty}e^{-\zeta_{k}x_{k}}\Big(\pi_{k}^{0}+\int_{0^{+}}^{x_{k}}e^{\zeta_{k}v_{k}}f_{k}(v_{k}) \d v_{k} \Big)\d x_{k}\\  &={\lambda_{k}\over \zeta_{k}}\pi_{k}^{0} +\lambda_{k}\int_{0^{+}}^{\infty}\int_{0^{+}}^{x_{k}}e^{\zeta_{k}(v_{k}-x_{k})}f_{k}(v_{k})\d v_{k}\d x_{k}\\ &\stackrel {\rm{(a)}}{=}{\lambda_{k}\over \zeta_{k}}\pi_{k}^{0}+\lambda_{k}\int_{0^{+}}^{\infty}\int_{v_{k}}^{\infty}e^{\zeta_{k}(v_{k}-x_{k})}f_{k}(v_{k})\d x_{k}\d v_{k}\\
&={\lambda_{k}\over \zeta_{k}}\pi_{k}^{0}+{\lambda_{k}\over \zeta_{k}}\int_{0^{+}}^{\infty}f_{k}(v_{k})\d v_{k}\\ \label{average-power-transmission}
&={\lambda_{k}\over \zeta_{k}},
\end{align}
where in \rm{(a)}, we changed the order of integration. Thus, for positive recurrent policies and when all $L_{k}=\infty$, we have the following upper bound
\begin{align}
\label{upper bound 1}
\sum_{k=1}^{M}R_{k} \leq r\Big(\sum_{k=1}^{M}{\lambda_{k}\over \zeta_{k}}\Big).
\end{align}
\begin{remark}
In contrast with the inequality \eqref{upper bound} which only holds for positive recurrent transmission power policies, Eq. \eqref{upper bound 1} is valid for transient and null recurrent power policies as well.
In particular, in the infinite battery case,
\begin{align}
\nonumber
\lim_{T\rightarrow \infty}\dfrac{1}{T} \int_{0}^{T}p_{k}\big(X_{k}(t)\big)\d t \leq {\lambda_{k}/\zeta_{k}},
\end{align}
regardless of the type of power policy, and thus \eqref{upper bound 1} follows by concavity of the rate function. Nevertheless, the strict equality in Eq. \eqref{average-power-transmission} will be used to study transmission power policies that result in ergodic behavior for the infinite battery capacity case in Section \ref{subsection:Achievable allocation scheme}.
\end{remark}
\subsection{Achievable allocation scheme}
\label{subsection:Achievable allocation scheme}
To derive transmission power policies with good performance, we start with the ergodicity assumption and the definition of expectation, \textit{i.e.},
\begin{align}
\sum_{k=1}^{M}R_{k}&=\lim_{T\rightarrow \infty} {1\over T}\int_{0}^{T}r\Big(\sum_{k=1}^{M}P_{k}(s)\Big)\d s\\
&\stackrel{\text{a.s.}}{=} \label{12} \int_{\mathcal{A}}r\Big(\sum_{k=1}^{M}p_{k}(x_{k})\Big)\prod_{k=1}^{M}\pi_{k}(\d x_{k})\\
&\triangleq \widehat{R}\big(\{p_{k}(x_{k})\}_{k=1}^{M}\big),
\end{align}
where
\begin{align}
\pi_{k}(\d x_{k})=[\pi_{k}^{0}\delta(x_{k})+f_{k}(x_{k})]\d x_{k},
\end{align}
and $\delta(x_{k})$ denotes the Dirac delta function. We now aim to find achievable policies through the following optimization problem
\begin{subequations}
\begin{align}
 \label{constraint-feasibility}
&\sup_{\{\pi_{k}^{0},f_{k}(x_{k})\}_{k=1}^{M}}\int_{\mathcal{A}}r\Big(\sum_{k=1}^{M}p_{k}(x_{k})\Big)\prod_{k=1}^{M}\pi_{k}(\d x_{k}),\\ &{\text{s.t. :}}  \label{positivity-feasibility-begin}
f_{k}(x_{k})={\lambda_{k}e^{-\zeta_{k} x} \over p_{k}(x_{k})}\left(\pi_{k}^{0}+\int_{0^{+}}^{x_{k}}e^{-\zeta_{k}v}f_{k}(v) \d v \right),\\
&\hspace{6mm}\pi_{k}^{0}+\int_{0^{+}}^{L_{k}}f_{k}(x_{k})\d x_{k} = 1,\\ \label{positivity-feasibility}
&\hspace{6mm} \pi_{k}^{0}\geq 0, \ \ f_{k}(x_{k})\geq 0, \quad \forall k\in [M],
\end{align}
\end{subequations}
which maximizes the overall expected throughput of the multiple access channel subject to the stationary probability measure constraints of the batteries. However, tackling this non-linear optimization problem is challenging as the feasibility constraint in Eq. \eqref{positivity-feasibility-begin} is not in an explicit form. To circumvent this difficulty, we use a calculus of variations approach to transform the problem into a set of necessary conditions for an optimal solution. As a starting point, consider the following linear mappings
\begin{align}
g_{k}(x_{k})\triangleq f_{k}(x_{k})e^{\zeta_{k}x_{k}}, \ \ \ x_{k}>0,
\end{align}
that transforms the positive recurrent condition in Eq. \eqref{Forward equation summarized form} into
\begin{align}
g_{k}(x_{k})&={\lambda_{k}\over p_{k}(x_{k})}\left(\pi_{k}^{0}+\int_{0^{+}}^{x_{k}}g_{k}(v) \d v \right)\\ \label{G-p}
&= {\lambda_{k}\over p_{k}(x_{k})}G_{k}(x_{k}),
\end{align}
with $G_{k}(x_{k})= \Big(\pi_{k}^{0}+\int_{0^{+}}^{x_{k}}g_{k}(v) \d v \Big)$ as in Eq. \eqref{F}. Hence, \eqref{12} is valid with
\begin{align}
p_{k}(x_{k})&=
\begin{cases}
\label{13}
{\lambda_{k}G_{k}(x_{k})/ g_{k}(x_{k})} & x_{k}> 0\\
0 &x_{k}=0,
\end{cases}
\\ \label{889}
\pi_{k}(\d x_{k})&=[\pi_{k}^{0}\delta(x_{k})+e^{-\zeta_{k}x_{k}}g_{k}(x_{k})]\d x_{k}.
\end{align}
With this substitution, we obtain an equivalent formulation for the optimization problem in Eqs. \eqref{constraint-feasibility}-\eqref{positivity-feasibility} as below
\begin{subequations}
\begin{align}
\label{constraint-feasibility1}
&\sup_{\{\pi_{k}^{0}\},\{g_{k}(x_{k})\}}\int_{\mathcal{A}}r\Big(\sum_{k=1}^{M}{p_{k}(x_{k})}\Big)\prod_{k=1}^{M}\pi_{k}(\d x_{k}),\\ &{\text{s.t. :}} \label{positivity-feasibility2}
G_{k}(x_{k})= \Big(\pi_{k}^{0}+\int_{0^{+}}^{x_{k}}g_{k}(v) \d v \Big),\\
&\hspace{6mm}\pi_{k}^{0}+\int_{0^{+}}^{L_{k}}e^{-\zeta_{k}v}g_{k}(v)\d v = 1,\\ \label{positivity-feasibility1}
&\hspace{6mm} \pi_{k}^{0}\geq 0, \ \ g_{k}(x_{k})\geq 0,\quad \forall k\in [M],
\end{align}
\end{subequations}
where $p_{k}(x_{k})$ and $\pi_{k}(\d x_{k})$ are according to \eqref{13} and \eqref{889}.

Through the formulation in Eqs. \eqref{constraint-feasibility1}-\eqref{positivity-feasibility1}, we can show that the throughput maximization problem in Eqs. \eqref{constraint-feasibility}-\eqref{positivity-feasibility} is concave with respect to each coordinate over a convex feasible set. In particular, since the transformation between $f_{k}(x_{k})$ and $g_{k}(x_{k})$ is linear, the concavity of \eqref{constraint-feasibility}-\eqref{positivity-feasibility} can be shown equivalently by proving the concavity of the formulation in Eqs. \eqref{constraint-feasibility1}-\eqref{positivity-feasibility1}. To this end, suppose that $\big\{\big(\pi_{k}^{0,1},g_{k}^{1}(x_{k})\big)\big\}_{k=1}^{M}\ \text{and} $ $\big\{\big({{\pi}}_{k}^{0,2},{g}_{k}^{2}(x_{k})\big)\big\}_{k=1}^{M}$ are two arbitrary sets of optimization parameters belonging to the feasible region defined in Eqs. \eqref{positivity-feasibility2}-\eqref{positivity-feasibility1}. Then for all $\alpha \in [0,1]$ and $\bar{\alpha}\triangleq(1-\alpha)$, it readily follows that $\{\big(\pi_{k}^{0,\alpha},g_{k}^{\alpha}(x_{k})\big)\}_{k=1}^{M}$ also satisfies \eqref{positivity-feasibility2}-\eqref{positivity-feasibility1}, where $\pi_{k}^{0,\alpha}=\alpha{{\pi}}_{k}^{0,1}+ \bar{\alpha}{{\pi}}_{k}^{0,2}$ and $g_{k}^{\alpha}(x_{k})=\alpha {g}_{k}(x_{k})+\bar{\alpha}{g}_{k}(x_{k})$ are the convex combinations of the densities and atoms, respectively. This proves the convexity of the feasible region \eqref{positivity-feasibility2}-\eqref{positivity-feasibility1}.
\begin{proposition}(\textit{Coordinate-wise Convexity})
\label{proposition:1}
Let $\widehat{R}^{\alpha}_{j}$, $\widehat{R}_{j}^{1}$ and $\widehat{R}_{j}^{2}$ be the utility functions corresponding to $\big\{\big(\pi_{k}^{0,\alpha},g_{k}^{\alpha}(x_{k})\big)\big\}_{k=1}^{M},$
$\big\{\big(\pi_{k}^{0,1},g_{k}^{1}(x_{k})\big)\big\}_{k=1}^{M}$, and $\big\{\big({\pi}_{k}^{0,2},{g}_{k}^{2}(x_{k})\big)\big\}_{k=1}^{M}$ respectively, such that
\begin{align}
\nonumber
&\big(\pi_{k}^{0,\alpha},g^{\alpha}_{k}(x_{k})\big)=\alpha \big(\pi_{k}^{0,1},g^{1}_{k}(x_{k})\big)+\bar{\alpha}\big({\pi}_{k}^{0,2},{g}^{2}_{k}(x_{k})\big),& k= j, \\ \nonumber
&\big(\pi_{k}^{0,\alpha},g^{\alpha}_{k}(x_{k})\big)=\big(\pi_{k}^{0,1},g^{1}_{k}(x_{k})\big)=\big({\pi}_{k}^{0,2},{g}^{2}_{k}(x_{k})\big),&k\neq j.
\end{align}
Then,
\begin{align}
\widehat{R}^{\alpha}_{j}\geq \alpha \widehat{R}_{j}^{1}+\bar{\alpha}{\widehat{R}}_{j}^{2}.
\end{align}
\begin{proof}
The proof is relegated to Appendix B.
\end{proof}
\end{proposition}

Now, define an ensemble of perturbation functions, $\{\psi_{k}\}_{k=1}^{M}$, such that
\begin{align}
\label{888}
\int_{0^{+}}^{L_{k}}\psi_{k}(v) \d v &=0\\
\label{constraint}
\int_{0^{+}}^{L_{k}}\exp(-\zeta_{k}v) \psi_{k}(v) \d v &=0,
\end{align}
and the $\psi_{k}$ are continuous and bounded on their domain $(0,L_{k}]$ and $\psi_{k}(0)=0$. For sufficiently small $\varepsilon_{k}>0, k \in [M]$, it thus follows that $g^{\varepsilon_{k}}_{k}(x_{k})\triangleq g_{k}(x_{k})+\varepsilon_{k} \psi_{k}(x_{k})$ satisfies \eqref{positivity-feasibility2}-\eqref{positivity-feasibility1} with the same atoms $\pi^{0}_{k}$ and thus lies inside the feasibility region. Then, with $\bm{\varepsilon} \triangleq (\varepsilon_{1},\varepsilon_{2},\cdots, \varepsilon_{M})$, it must be true for a global maximum solution that
\begin{align}
\label{StationaryPointInequality}
\widehat{R}^{\bm{\varepsilon}}\leq \widehat{R},
\end{align}
where
\begin{align}
\label{R(s)}
{\widehat{R}}^{\bm{\varepsilon}}= \int_{\mathcal{A}}r\Big(\sum_{k=1}^{M}p_{k}^{\varepsilon_{k}}(x_{k})\Big)\prod_{k=1}^{M}\pi_{k}^{\varepsilon_{k}}(\d x_{k}),
\end{align}
and
\begin{align}
\nonumber
\pi_{k}^{\varepsilon_{k}}(x_{k})&\triangleq [\pi_{k}^{0}\delta(x_{k})+e^{-\zeta_{k}x_{k}}g_{k}(x_{k})+\varepsilon_{k}e^{-\zeta_{k}x_{k}}\psi_{k}(x_{k})]\d x_{k}\\ &=\pi_{k}(\d x_{k})+\varepsilon_{k}e^{-\zeta_{k}x_{k}}\psi_{k}(x_{k})\d x_{k},
\end{align}
and $p_{k}^{\varepsilon_{k}}(x_{k})$ is calculated from Eq. \eqref{13} to be,
\begin{align}
\label{pp}
p_{k}^{\varepsilon_{k}}(x_{k})=
\begin{cases}
{\lambda_{k}}\dfrac {{G_{k}(x_{k})+\varepsilon_{k}\Psi_{k}(x_{k})}}{{g_{k}(x_{k})+\varepsilon_{k}\psi_{k}(x_{k})}}  & x_{k}> 0 \\
0  &x_{k}=0,
\end{cases}
\end{align}
with,
\begin{align}
\label{Eq:definitionofPsi}
\Psi_{k}(x_{k})\triangleq \int_{0}^{x_{k}}\psi_{k}(v) \d v.
\end{align}
For the moment, we assume that only the $j^{th}$ coordinate is perturbed; that is $\varepsilon_{k}=0, \forall k\neq j$. Expanding the right hand side of \eqref{R(s)} to first order then results in
\small\begin{align}
\nonumber
{\widehat{R}}^{\varepsilon_{j}} &=\int_{\mathcal{A}}\Bigg[r\Big(\sum_{k=1}^{M}p_{k}(x_{k})\Big)+\varepsilon_{j}{\partial r\Big(\sum_{k=1}^{M}p_{k}(x_{k})\Big)  \over \partial p_{j}(x_{j})}{\d p^{\varepsilon_{j}}_{j}(x_{j})\over \d\varepsilon_{j}}\Big \vert_{\varepsilon_{j}=0}\Bigg] \\ \nonumber &\times \Big[\pi_{j}(\d x_{j})+\varepsilon_{j}e^{-\zeta_{j}x_{j}}\psi_{j}(x_{j})\d x_{j}\Big]\hspace{-1mm}\prod_{k\in [M]-j}\hspace{-2mm}\pi_{k}(\d x_{k})\\ \nonumber
&=\widehat{R}+\varepsilon_{j} \int_{\mathcal{A}}r\Big(\sum_{k=1}^{M}p_{k}(x_{k})\Big)e^{\zeta_{j}x_{j}}\psi_{j}(x_{j})\d x_{j}\hspace{-2mm}\prod_{k\in [M]-j}\hspace{-3mm}\pi_{k}(\d x_{k})\\ \label{14}
&+\varepsilon_{j}\int_{\mathcal{A}}{\partial r\big(\sum_{k=1}^{M} {p_{k}(x_{k})}\big) \over \partial {p_{j}(x_{j})}}{\d p_{j}^{\varepsilon_{j}}(x_{j})\over \d\varepsilon_{j}}\Big \vert_{\varepsilon_{j} =0}  \prod_{k=1}^{M}\pi_{k}(\d x_{k})+\mathcal{O}(\varepsilon_{j}^{2}).
\end{align}\normalsize
On the other hand, we note that
\begin{align}
\label{15}
{\d p_{j}^{\varepsilon_{j}}(0)\over \d \varepsilon_{j}}\Big \vert_{\varepsilon_{j}=0} =0,
\end{align}
since $p_{k}^{\varepsilon_{j}}(0)=0$ from \eqref{pp}.\footnote{Alternatively, $p_{j}^{\varepsilon_{j}}(0)=0$ for all $\varepsilon_{j}$ as the battery is empty.} Therefore,
\begin{align}
\int_{\mathcal{A}}{\partial r\big(\sum_{k=1}^{M} {p_{k}(x_{k})}\big) \over \partial {p_{j}(x_{j})}}{\d p_{j}^{\varepsilon_{j}}(x_{j})\over \d\varepsilon_{j}}\Big \vert_{\varepsilon_{j} =0} \delta(x_{j})\d x_{j}=0,
\end{align}
and thus from Eq. \eqref{14} we obtain \eqref{eq:longline1} on the next page.
\begin{figure*}[!t]
\normalsize

\begin{align}
\nonumber
{\widehat{R}}^{\varepsilon_{j}}&= \widehat{R}+\varepsilon_{j} \int_{\mathcal{A}}r\Big(\sum_{k=1}^{M}p_{k}(x_{k})\Big)e^{\zeta_{j}x_{j}}\psi_{j}(x_{j})\d x_{j}\hspace{-2mm}\prod_{k\in [M]-j}\hspace{-2mm}\pi_{k}(\d x_{k})\\ \label{eq:longline1}
&\hspace{4mm}+\varepsilon_{j}\int_{\mathcal{A}}{\partial r\big(\sum_{k=1}^{M} {p_{k}(x_{k})}\big) \over \partial {p_{j}(x_{j})}}{dp_{j}^{\varepsilon_{j}}(x_{j})\over d\varepsilon_{j}}\Big \vert_{\varepsilon_{j} =0}e^{\zeta_{j}x_{j}} g_{j}(x_{j})\d x_{j} \hspace{-2mm} \prod_{k\in[M]-j}\pi_{k}(\d x_{k})+\mathcal{O}(\varepsilon_{j}^{2}).
\end{align}
\hrulefill
\begin{align}
\nonumber
&\int_{\mathcal{A}}r\Big(\sum_{k=1}^{M}p_{k}(x_{k})\Big)e^{\zeta_{j}x_{j}}\psi_{j}(x_{j})\d x_{j}\prod_{k\in [M]-j}\pi_{k}(\d x_{k})\\ \label{necessary-condition}&+\int_{\mathcal{A}}{\partial r\big(\sum_{k=1}^{M} {p_{k}(x_{k})}\big) \over \partial {p_{j}(x_{j})}}{dp_{j}^{\varepsilon_{j}}(x_{j})\over d\varepsilon_{j}}\Big \vert_{\varepsilon_{j} =0} e^{\zeta_{j}x_{j}}g_{j}(x_{j})\d x_{j}\prod_{k\in [M]-j}\pi_{k}(\d x_{k})=0.
\end{align}
\hrulefill
\setcounter{equation}{73}
\begin{align}
\nonumber
&\int \limits _{0}^{L_{j}}\Big[ e^{-\zeta_{j}x_{j}}{\partial \expect_{j}\big[r\big(\sum_{k=1}^{n} {p_{k}(x_{k})}\big)\big] \over \partial {p_{j}(x_{j})}}p_{j}(x_{j})+e^{-\zeta_{j}x_{j}}\expect_{j}\big[r\big(\sum_{k=1}^{M} {p_{k}(x_{k})}\big)\big] \Big]\psi_{j}(x_{j})\d x_{j}\\ \nonumber &= \left(e^{-\zeta_{j}x_{j}} {\partial \expect_{j}\big[r\big(\sum_{k=1}^{M} {p_{k}(x_{k})}\big)\big] \over \partial {p_{j}(x_{j})}}p_{j}(x_{j})+e^{-\zeta_{j}x_{j}}\expect_{j}\big[r\big(\sum_{k=1}^{M} {p_{k}(x_{k})}\big)\big] \right)\Psi_{j}(x_{j})\Big \vert_{0}^{L_{j}}\\ \label{expression} &\hspace{4mm}+\int \limits _{0}^{L_{j}}{\partial \over \partial x_{j}} \Big[ e^{-\zeta_{j}x_{j}}{\partial \expect_{j}\big[r\big(\sum_{k=1}^{M} {p_{k}(x_{k})}\big)\big] \over \partial {p_{j}(x_{j})}}p_{j}(x_{j}) +
e^{-\zeta_{j}x_{j}}\expect_{j}\big[r\big(\sum_{k=1}^{M} {p_{k}(x_{k})}\big)\big] \Big]\Psi_{j}(z_{j})\d x_{j}.
\end{align}
\hrule
\begin{align}
\nonumber
&\int \limits _{0}^{L_{j}}\Bigg( \lambda_{j}e^{-\zeta_{j}x_{j}}{\partial \expect_{j}\big[r\big(\sum_{k=1}^{M} {p_{k}(x_{k})}\big)\big] \over \partial {p_{j}(x_{j})}}\\ \label{family-solution} &\hspace{20mm}-{\partial \over \partial x_{j}} \Big[ e^{-\zeta_{j}x_{j}}{\partial \expect_{j}\big[r\big(\sum_{k=1}^{M} {p_{k}(x_{k})}\big)\big] \over \partial {p_{j}(x_{j})}}p_{j}(x_{j}) +
e^{-\zeta_{j}x_{j}}\expect_{j}\big[r\big(\sum_{k=1}^{M} {p_{k}(z_{k})}\big)\big] \Big]\Bigg) \Psi_{j}(x_{j})\d x_{j}=0.
\end{align}
\hrule
\vspace*{4pt}
\end{figure*}
This expansion, accompanied with inequality \eqref{StationaryPointInequality} establishes \eqref{necessary-condition} as a necessary condition for a locally (and thus globally) optimal solution, where we have neglected the second order term $\mathcal{O}(\varepsilon_{j}^{2})$. Now with slight abuse of notation, let
\setcounter{equation}{67}
\begin{align}
\label{expectation-j}
&\expect_{j}\Bigg[r\Big(\sum_{k=1}^{M} {p_{k}(x_{k})}\Big)\Bigg]\triangleq
\int_{\mathcal{A}_{j}}r\Big(\sum_{k=1}^{M} {p_{k}(x_{k})}\Big) \prod_{{k\in [M]-j}}\pi_{k}(\d x_{k}),
\end{align}
denote the expectation over all the coordinates except the $j$-th coordinate. Then \eqref{necessary-condition} can be restated as
\begin{align}
\nonumber
&\int \limits _{0}^{L_{j}}\Bigg[ {\partial \expect_{j}\Big[r\big(\sum_{k=1}^{n} {p_{k}(x_{k})}\big)\Big] \over \partial {p_{j}(x_{j})}}{\d p_{j}^{\varepsilon_{j}}(x_{j})\over \d\varepsilon_{j}}\Big \vert_{\varepsilon_{j} =0}e^{-\zeta_{j}x_{j}}g_{j}(x_{j})  \\  \label{expectationform} &\hspace{6mm}+\expect_{j}\Big[r\big(\sum_{k=1}^{M} {p_{k}(z_{k})}\big)\Big]e^{-\zeta_{j}x_{j}} \psi_{j}(z_{j})\Bigg]\d x_{j}=0,
\end{align}
where we used the fact that
\begin{align}
\nonumber
{\expect}_{j}\left[{\partial r\Big(\sum_{i=1}^{n} {p_{i}(z_{i})}\Big) \over \partial {p_{j}(z_{j})}}\right]={{\partial \over \partial {p_{j}(z_{j})}} {\expect}_{j} \left[r\Big(\sum_{i=1}^{n} {p_{i}(z_{i})}\Big)\right]}.
\end{align}
On the other hand, from Eq. \eqref{pp}, we compute
\begin{align}
g_{j}(x_{j}){\d p^{\varepsilon_{j}}_{j}(x_{j})\over \d\varepsilon_{j}}\Big \vert_{\varepsilon_{j}=0}&=\lambda_{j}\Big[\Psi_{j}(x_{j})-{\psi_{j}(x_{j})G_{j}(x_{j})\over g_{j}(x_{j})}\Big]\\ \label{replace} &=\lambda_{j}\Psi_{j}(x_{j})-\psi_{j}(x_{j})p_{j}(x_{j}).
\end{align}
We thus further proceed by substituting \eqref{replace} in Eq. \eqref{expectationform}, \textit{i.e.},
\begin{align}
\nonumber
&\int \limits _{0}^{L_{j}}\Big[ \lambda_{j}e^{-\zeta_{j}x_{j}}{\partial \expect_{j}\big[r\big(\sum_{k=1}^{M} {p_{k}(x_{k})}\big)\big] \over \partial {p_{j}(x_{j})}}\Big]\Psi_{j}(x_{j})\d x_{j}\\ \nonumber & -
\int \limits _{0}^{L_{j}}\Big[ e^{-\zeta_{j}x_{j}}{\partial \expect_{j}\big[r\big(\sum_{k=1}^{M} {p_{k}(x_{k})}\big)\big] \over \partial {p_{j}(x_{j})}}p_{j}(x_{j})\\  \label{equivalent}
&+e^{-\zeta_{j}x_{j}}\expect_{j}\big[r\big(\sum_{k=1}^{M} {p_{k}(x_{k})}\big)\big] \Big]\psi_{j}(x_{j})\d x_{j}=0.
\end{align}
Integrating by parts, the second integral can be evaluated as \eqref{expression}. But since $\Psi_{j}(0)=\Psi_{j}(L_{j})=0$ due to the definition in Eq. \eqref{Eq:definitionofPsi} and \eqref{888}, the second term in Eq. \eqref{expression} vanishes. Replacing the remaining terms in Eq. \eqref{equivalent} yields \eqref{family-solution} which must hold for all defined $\Psi_{j}(x_{j})$ in Eq. \eqref{Eq:definitionofPsi}, such that $\psi_{j}(x_{j})$ satisfies both \eqref{888} and \eqref{constraint}. A family of solutions for Equation \eqref{family-solution} can be supplied by simultaneously noting from Eq. \eqref{constraint} that
\setcounter{equation}{75}
\begin{align}
0&=\int_{0}^{L_{j}}e^{-\zeta_{j}x_{j}}\psi_{j}(x_{j})\d x_{j}
\\&=e^{-\zeta_{j}x_{j}}\Psi_{j}(x_{j})\Big \vert_{0}^{L_{j}}+\zeta_{j} \int_{0}^{L_{j}}e^{-\zeta_{j}x_{j}}\Psi_{j}(x_{j})\d x_{j}
\\ \label{zero-equality}&\stackrel {\rm{(a)}}{=} \zeta_{j} \int_{0}^{L_{j}}e^{-\zeta_{j}x_{j}}\Psi_{j}(x_{j})\d x_{j},
\end{align}
where \rm{(a)} is true since $\Psi_{j}(0)=\Psi_{j}(L_{j})=0$ as noted before. Now, if $p_{j}(x_{j})$ is twice continuously differentiable, it follows that the term inside the parentheses in Eq. \eqref{family-solution} is continuously differentiable as the rate function $r(x)$ is assumed to be three times continuously differentiable. Furthermore, since $\Psi_{j}(x_{j})$ is also continuously differentiable, from \eqref{zero-equality} and the fundamental lemma of the calculus of variations, we then conclude that \eqref{family-solution} holds only if the term inside the parentheses is in the form of $K_{j}\exp(-\zeta_{j}x_{j})$ for some constant $K_{j}$. Thus, as a necessary condition, we have
\begin{align}
\label{necessary-condition1}
&\nonumber {p_{j}(x_{j}){p'_{j}}(x_{j}){\partial^{2} {\expect}_{j}[r(\sum_{k=1}^{M}{p_{k}(x_{k})})]\over \partial^{2} {p_{j}(x_{j})}}}+(\lambda_{j}-\zeta_{j} p_{j}(x_{j}))\\  &\times {\partial {\expect}_{j}[r(\sum_{k=1}^{M} {p_{k}(x_{k})})]\over \partial {p_{j}(x_{j})}}+\zeta_{j} {\expect}_{j}[r(\sum_{k=1}^{M} {p_{k}(x_{k})})]+K_{j}=0.
\end{align}
\begin{remark} Rewriting Equation \eqref{necessary-condition1} as
\small\begin{align}
&{p'_{j}}(x_{j})= 
-\nonumber {p_{j}(x_{j}){\partial^{2} {\expect}_{j}[r(\sum_{k=1}^{M}{p_{k}(x_{k})})]\over \partial^{2} {p_{j}(x_{j})}}}\big]^{-1}\big[(\lambda_{j}-\zeta_{j} p_{j}(x_{j}))\\ \label{initialize-K} &\times {\partial {\expect}_{j}[r(\sum_{k=1}^{M} {p_{k}(x_{k})})]\over \partial {p_{j}(x_{j})}}+\zeta_{j} {\expect}_{j}[r(\sum_{k=1}^{M} {p_{k}(x_{k})})]+K_{j},
\end{align}\normalsize
it is easy to verify that $K_{j}$ provides a degree of freedom to set the initial slope $p_{j}'(x_{j})\big \vert_{x_{j}=0^{+}}$ of the power policy $p_{j}(x_{j})$.
\end{remark}

Now since the choice of $j^{th}$ coordinate was arbitrary, \eqref{necessary-condition1} holds for all $j\in [M]$. Accordingly, we obtain a system of coupled PIEDs over $M$ coordinates with $2M$ degree of freedom\footnote{In the system of equations, $\{p_{k}(0^{+})\}_{k=1}^{M}$ and $\{p'_{k}(0^{+})\}_{k=1}^{M}$ (or equivalently $\{p_{k}(0^{+})\}_{k=1}^{M}$ and $\{K_{k}\}_{k=1}^{M}$) are free parameters.}  where the integration is implicit in the notation of $\expect_{j}[\cdot]$ (ref. Eq. \eqref{expectation-j}). In the following, we consider solutions in the infinite and finite battery cases.
\subsubsection{Infinite Storage Battery} Motivated by the converse result for the infinite storage battery, we consider a set of admissible policies as below
\begin{align}
\label{policies-infinite}
\bar{p}_{k}(x_{k}) = \begin{cases}
 {\big(\lambda_{k}/\zeta_{k}\big)}+\varrho & \text{$x_{k}>0$} \\
  0 & \text{$x_{k}=0$},
\end{cases}
\end{align}
where the excess power $\varrho>0$ is added to ensure the positive recurrence of the process. In the limit, as $\varrho \rightarrow 0$, the suggested policies in Eq. \eqref{policies-infinite} satisfy \eqref{necessary-condition1} for all $j$, provided
\begin{align}
K_{j}= -\zeta_{j}{\expect}_{j}\left[r\big(\sum_{k=1}^{M} {\lambda_{k}\over \zeta_{k}}\big)\right].
\end{align}
The average transmission power of \eqref{policies-infinite} can then be evaluated as
\begin{align}
\hspace{-4mm}\expect[\bar{p}_{k}(x_{k})]&=\pi_{k}^{0}\bar{p}_{k}(0)+\big((\lambda_{k}/\zeta_{k})+\varrho\big)\int_{0^{+}}^{\infty}f_{k}(x_{k})\d x_{k}\\
&=\big((\lambda_{k}/\zeta_{k})+\varrho \big)\big(1-\pi_{k}^{0}\big).
\end{align}
On the other hand, in light of Eq. \eqref{average-power-transmission} we have $\expect[\bar{p}_{k}(x_{k})]=\lambda_{k}/\zeta_{k}$. As a result, a transmission node that exploits $\bar{p}_{k}(x_{k})$ as transmission power policy has the following probability mass at zero
\begin{align}
\pi_{k}^{0}={\varrho \over (\lambda_{k}/\zeta_{k})+\varrho},
\end{align}
\textit{i.e.}, it sends information for a fraction $\dfrac{(\lambda_{k}/\zeta_{k})}{(\lambda_{k}/\zeta_{k})+\varrho}$ of time. Moreover, associated with $\bar{p}_{k}(x_{k})$, the mean square deviation of transmission power is given by
\begin{align}
\nonumber
\sigma^{2}\big(\bar{p}_{k}(X_{k})\big)&=\big({{\lambda_{k}\over \zeta_{k}}}\big)^{2}\pi_{k}^{0}+\int_{0^{+}}^{\infty}\Big(\bar{p}_{k}(x_{k})-{\lambda_{k}\over \zeta_{k}}\Big)^{2}f_{k}(x_{k}) \ d {x_{k}}
\\&={{(\lambda_{k}/ \zeta_{k})}{\varrho}}.
\end{align}
For these power transmission strategies, we also have
\begin{align}
\label{lowlow}
\widehat{R}\geq r\Big(\sum_{k=1}^{M}(\lambda_{k}/\zeta_{k}+\varrho)\Big)\prod_{k=1}^{M}{{(\lambda_{k}/ \zeta_{k})}\over {(\lambda_{k}/ \zeta_{k})+{\varrho}}},
\end{align}
where the inequality follows from neglecting situations in which a strict subset of nodes are transmitting and the rest are silent due to battery exhaustion.  As ${\varrho \downarrow 0}$, the upper bound \eqref{upper bound 1} and the lower bound \eqref{lowlow} coincide with each other. The total average throughput is given asymptotically by
\begin{align}
\widehat{R} = r\Big(\sum_{k=1}^{M}\lambda_{k}/\zeta_{k}\Big).
\end{align}
\begin{figure*}[!b]
\normalsize
\hrule
\setcounter{equation}{95}
\begin{align}
\label{eq:meusamline}
&p_{j}^{(N+1)}(x_{j})=\arg \max_{\xi}\widehat{R}(p_{1}^{(N+1)}(x_{1}),\cdots, p_{j-1}^{(N+1)}(x_{j-1}),\xi, p_{j+1}^{(N)}(x_{j+1}),\cdots, p_{M}^{(N)}(x_{M})).
\end{align}
\end{figure*}
Thus, near optimal performance of the energy harvesting system can be achieved when $\varrho\downarrow 0$, and the behaviour of a classical communication systems (in the sense of using a constant power supply without interruption) can be closely approximated at the same time.
\subsubsection{Finite Storage Battery}
In contrast to the case of batteries with infinite capacity, the system of equations in Eq. \eqref{necessary-condition1} doesn't appear to admit a closed form expression for power policies when the storage capacity is finite. The remaining option is thus to solve \eqref{necessary-condition1} numerically. However, the complexity in dealing with such systems is that the equations are not independent, but \textit{coupled}. Alternatively, if all but one of the $p_{k}(\cdot)$ are known, the remaining one can be obtained by solving a first order non-linear ODE in terms of the corresponding coordinate, using \eqref{necessary-condition1}. First from \eqref{G-p} and since $G'_{k}(x_{k})=g_{k}(x_{k})$ we obtain by integration
\setcounter{equation}{89}
\begin{align}
\ln G_{k}(x_{k})-\ln G_{k}(0)=\int_{0^{+}}^{x_{k}}{\lambda_{k}\over p_{k}(v_{k})}\d v_{k}.
\end{align}
Therefore,
\begin{align}
G_{k}(x_{k})=\pi_{k}^{0}\exp \Big(\int_{0^{+}}^{x_{k}}{\lambda_{k}\over p_{k}(v)}\d v \Big).
\end{align}
By differentiating both sides with respect to $x_{k}$, we obtain
\begin{align}
g_{k}(x_{k})={\pi_{k}^{0}\lambda_{k}\over p_{k}(x_{k})}\exp \Big(\int_{0^{+}}^{x_{k}}{\lambda_{k}\over p_{k}(v)}\d v \Big),
\end{align}
or equivalently,
\begin{align}
\label{asset-1}
f_{k}(x_{k})=\pi_{k}^{0}{{e^{-\lambda_{k}x_{k}}}\lambda_{k}\over p_{k}(x_{k})}\exp \Big(\int_{0^{+}}^{x_{k}}{\lambda_{k}\over p_{k}(v)}\d v \Big).
\end{align}
Also, due to the normalization condition
\begin{align}
\label{asset-2}
\pi_{k}^{0}+\int_{0^{+}}^{L_{k}}f_{k}(x_{k})\d x_{k}=1,
\end{align}
we have
\small \begin{align}
\label{asset-3}
\pi_{k}^{0}=\left[1+\int_{0^{+}}^{L_{k}}{e^{-\lambda_{k}x_{k}}\lambda_{k}\over p_{k}(x_{k})}\exp \Big(\int_{0^{+}}^{x_{k}}{\lambda_{k}\over p_{k}(v)}\d v \Big)\d x_{k} \right]^{-1},
\end{align}\normalsize
which can simply be derived via substituting \eqref{asset-1} in Eq. \eqref{asset-2} and solving for $\pi_{k}^{0}$.

With the help of \eqref{asset-1} and \eqref{asset-3}, we propose an iterative method, outlined as Algorithm \ref{Alg:1}, that computes a solution for \eqref{necessary-condition1}. Also, the convergence analysis of Algorithm \ref{Alg:1} follows the following three steps:
\begin{itemize}
\item [(\textit{i})]  at each iteration of Algorithm \ref{Alg:1}, the utility function \eqref{12} is non-decreasing,
\item [(\textit{ii})] the utility function in Eq. \eqref{12} is bounded above,
\item [(\textit{iii})] the utility \eqref{12} thus converges if Algorithm \ref{Alg:1} is allowed to iterate indefinitely (\textit{i.e.} no termination constraint).
\end{itemize}

 Specifically, in the first step, we denote the utility $\widehat{R}$ as an explicit function of the power policies in Algorithm \ref{Alg:1}, \textit{e.g.},
$\widehat{R}(p_{1}^{(0)}(x_{1}),p_{2}^{(0)}(x_{2}),\cdots ,p_{M}^{(0)}(x_{M}))$ is the initial utility. After the $N$th full iteration of steps 5-10 (outer loop) of Algorithm \ref{Alg:1}, in the $j$th iteration of 6-10 (inner loop), we then obtain Eq. \eqref{eq:meusamline} below. Therefore,
\setcounter{equation}{96}
\begin{align}
\nonumber
&\widehat{R}\big(p_{1}^{(N+1)}(x_{1}),\cdots, p_{j}^{(N)}(x_{j}),\cdots, p_{M}^{(N)}(x_{M})\big)
 \\ \label{inequality1}
&\leq\widehat{R}\big(p_{1}^{(N+1)}(x_{1}),\cdots, p_{j}^{(N+1)}(x_{j}),\cdots, p_{M}^{(N)}(x_{M})\big).
\end{align}
In addition, since the objective function is upper bounded by \eqref{upper bound}, we further have
\begin{align}
\widehat{R}_{\text{sup}} &\triangleq \sup_{\{p_{k}(x_{k})\}} \widehat{R}(p_{1}(x_{1}),p_{2}(x_{2}),\cdots, p_{M}(x_{M}))\\
 \label{inequality2} &\leq r\Big(\sum_{k=1}^{M}{\lambda_{k}\over \zeta_{k}}(1-\exp(-\zeta_{k}L_{k}))\Big).
\end{align}
Concluding from \eqref{inequality1} and \eqref{inequality2}, the sequence $\widehat{R}\big(\{p_{k}^{N}(x_{k})\}_{N=0}^{\infty}\big)$ must converge in the limit as $N\rightarrow \infty$.

Now, we concentrate on two important degenerate cases of our problem that can substantially reduce the computational burden of solving the PIEDs.
\begin{algorithm}[t]
\caption{\small{Gauss-Seidel Alg. for Transmission Power Policies}}
\begin{algorithmic}[1]
\ForAll {$k \in [M]$}
\State Initialize $p_{k}(x_{k})$ with some arbitrary function;
\State compute \eqref{asset-1} and \eqref{asset-3};
\EndFor
\Repeat
\For{$j\gets 1, M$}
\State calculate \eqref{expectation-j};
\State update $p_{j}(x_{j})$ by solving \eqref{necessary-condition1} for optimized values of $p_{j}(0^{+})$ and $K_{j}$;
\State update \eqref{asset-1} and \eqref{asset-3} for $k=j$;
\EndFor
\Until{termination criterion is satisfied}.
\end{algorithmic}
\label{Alg:1}
\end{algorithm}
In the first scenario, suppose that all the transmission nodes scavenge energy in the same manner. By this statement, we mean that the statistical parameters of all the energy harvesters are identical, \textit{i.e.}, $\lambda_{k}=\lambda$ and $ \zeta_{k}=\zeta$ for all $k=1,2,\cdots, M$. In the symmetric case, further assume that the batteries have identical capacities ($L_{k}=L$) and all transmitters employ the same power policy $(p_{k}(x_{k})=p(x_{k}))$. Then, Equation \eqref{necessary-condition1} reduces to
\begin{align}
&\nonumber {p(x_{j}){p'}(x_{j}){\partial^{2} {\expect}_{j}[r(\sum_{k=1}^{M}{p(x_{k})})]\over \partial^{2} {p(x_{j})}}}+(\lambda-\zeta p(x_{j}))\\ \label{symm-MAC}  &\times {\partial {\expect}_{j}[r(\sum_{k=1}^{M} {p(x_{j})})]\over \partial {p(x_{j})}}+\zeta {\expect}_{j}[r(\sum_{k=1}^{M} {p(x_{j})})]+K=0,
\end{align}
where $j$ is arbitrary and chosen from $[M]$, and the operator $\expect_{j}[\cdot]$ now simplifies as
\begin{align}
\label{asset-4}
\expect_{j}\left[r\Big(\sum_{k=1}^{M} {p(x_{k})}\Big)\right] =
\int_{\mathcal{A}_{j}}r\Big(\sum_{k=1}^{M} {p(x_{k})}\Big) \prod_{{k\in [M]}-j}\pi(\d x_{k}).
\end{align}
If we rearrange the terms in Equation \eqref{symm-MAC}, we have that
\begin{align}
\label{fixed-point}
{p}(x_{j})= \mathscr{F}\big({p}(x_{j})\big),
\end{align}
where the mapping $\mathscr{F}(\cdot):\mathbb{C}^{1}(0,L]\rightarrow \mathbb{C}^{1}(0,L]$ is given by
\begin{align}
\nonumber
&\mathscr{F}\big({p}(x_{j})\big) = \\ \nonumber &p(0^{+})-\int_{0^{+}}^{x_{j}}\Big[K_{j}+\zeta {\expect}_{j}[r(\sum_{k=1}^{M} {p(v_{k})})]+ (\lambda-\zeta p(v_{j})) \\ \label{fix-point} &\times{\partial {\expect}_{j}[r(\sum_{k=1}^{M} {p(v_{k})})]\over \partial {p(v_{j})}}\Big]\Big[p(v_{j}){{\partial^{2} \expect}_{j}[r(\sum_{k=1}^{M}{p(v_{k})})]\over \partial^{2} {p(v_{j})}} \Big]^{-1}\d v_{j}.
\end{align}
As a result, it follows that the desired $p(x_{j})$ is a fixed point of $\mathscr{F}(\cdot)$. This then suggests an alternative algorithm for this special case (see Algorithm 2).

Now in the second scenario, consider that there is only one transmitter in the communication system (\textit{i.e.} a point-to-point setup). We thus have a simplified formulation as a necessary condition here, \textit{i.e.},
\small\begin{align}
\label{ODE}
&{p(x){p'}(x){\text{d}^{2} r\big(p(x)\big)\over \text{d}^{2} {p(x)}}}+\big(\lambda-\zeta p(x)\big) {\d r\big(p(x)\big)\over \d {p(x)}}+\zeta {r\big(p(x)\big)}+K=0.
\end{align}\normalsize
As argued in \cite{Mitran}, this is a second order, non-linear, autonomous ODE that can be solved numerically by employing linear multistep methods (\textit{e.g.} Runge-Kutta or Adams-Bashforth). The next lemma demonstrates some properties of solutions to this ODE.
\begin{lemma}
\label{Lemma:1}
Suppose $K>-\zeta r(\lambda/\zeta)$ in Eq. \eqref{ODE}, then for the Shannon rate function
\begin{enumerate}
\item[(\textit{i})] any solution $p(x)$ is a strictly increasing function of $x$ for $x\geq 0$, and $p(x)\rightarrow \infty$ as $x\rightarrow \infty$,
\item[(\textit{ii})] $p(x)$ grows doubly exponentially fast as $x\rightarrow \infty$.\hfill$\square$
\end{enumerate}
\end{lemma}
\begin{proof}
Solving \eqref{ODE} for $p'(x)$, we have
\begin{align}
\label{101}
p'(x)={\big(\lambda-\zeta p(x)\big)r'\big(p(x)\big)+\zeta {r\big(p(x)\big)}+K\over -{p(x)r''\big(p(x)\big)}}.
\end{align}
\begin{algorithm}[t]
\caption{Fixed Point Alg. For the Symmetric MAC}
\begin{algorithmic}[1]
\State \textbf{Initialize} $p^{(0)}(x_{j})$ with some function.
\Repeat
\State compute \eqref{asset-1} and \eqref{asset-3};
\State compute \eqref{asset-4};
\State update $p^{(N+1)}(x_{j})=\mathscr{F}(p^{(N)}(x_{j}))$ from Eq. \eqref{fix-point} for optimized values of $p^{(N+1)}(0^{+})$ and $K^{(N+1)}$;
\Until{termination criterion is satisfied}.
\end{algorithmic}
\label{Alg:2}
\end{algorithm}
From concavity of the rate function as well as the first constraint on admissible power policies we have $r''(p(x))< 0$ and $p(x)\geq 0$, respectively. Therefore, the denominator is always positive and $p'(x)>0$ for all $x\geq 0$ iff
\begin{align}
\hspace{-2mm}\label{K} K>-\big[\big(\lambda-\zeta p(x)\big)r'\big(p(x)\big)+\zeta {r\big(p(x)\big)}\big], \quad \forall x\geq 0.
\end{align}
Moreover, it can be verified that
\small\begin{align}
\nonumber
{\text{d}  \over \d p(x)}&\left[\big(\lambda-\zeta p(x)\big)r'\big(p(x)\big) +\zeta {r\big(p(x)\big)}\right]=\big(\lambda -\zeta p(x)\big)r''(p(x)).
\end{align}\normalsize
Hence, $p(x)=\lambda/\zeta$ is a global maxima for the right hand side of \eqref {K}. Replacing $p(x)=\lambda/\zeta$ into \eqref{K}, the numerator of \eqref{101} is then lower bounded by
\begin{align}
\nonumber
K+\zeta r(\lambda/\zeta)>0.
\end{align}
Furthermore, since $r(x)={1\over 2}\log_{2}\left(1+\dfrac{x}{N_{0}}\right)$, we upper bound the denominator by
\begin{align}
\nonumber
-p(x)r''(p(x))&=\frac{1/N_{0}}{2\ln 2}\frac{p(x)/N_{0}}{(1+p(x)/N_{0})^{2}}\\ \nonumber
&\leq \frac{1/N_{0}}{8\ln 2}.
\end{align}
and thus $p(x)\rightarrow +\infty$.

To prove the second part of Lemma \ref{Lemma:1}, consider the substitution
\begin{align}
p(x)/N_{0}=\exp(S(x)),
\end{align}
where $S(x)$ increases since $p(x)$ increases. Then we have
\begin{align}
\label{biggS}
S(x)\rightarrow +\infty,\ \ \text{as}\  x\rightarrow \infty.
\end{align}
Consequently, for the Shannon rate function we obtain
\begin{align}
\label{s1}
&{dr\big(p(x)\big)\over dp(x)} = {1\over {2\ln 2}}{1/N_{0}\over 1+p(x)/N_{0}} \simeq {1/N_{0}\over {2\ln 2}}{\exp(-S(x))},\\ \nonumber
&{d^{2}r\big(p(x)\big)\over dp(x)^{2}} = {-1\over {2\ln 2}}{1/N_{0}^{2}\over (1+p(x)/N_{0})^{2}}\\ &\hspace{15.8mm}\simeq {-1/N_{0}^{2}\over {2\ln 2}}{\exp(-2S(x))},
\end{align}
and
\begin{align}
\label{s2}
r(p(x))={1\over 2}\log_{2}\big(1+p(x)/N_{0}\big)\simeq {1\over 2\ln 2}S(x).
\end{align}
Replacing \eqref{s1}-\eqref{s2} in Eq. \eqref{ODE} yields that for $x\rightarrow \infty$
\begin{align}
\nonumber
&-S'(x)+{(\lambda-\zeta N_{0} e^{S(x)})}(e^{-S(x)}/N_{0}) \\ \label{s-eq2} &\hspace{30mm}+\zeta { S(x)}+{(K/2\ln 2)}=0,
\end{align}
As $x\rightarrow \infty$ and due to  \eqref{biggS}, Equation \eqref{s-eq2} reduces to
\begin{align}
\zeta S(x)=S'(x),
\end{align}
which has the following solution
\begin{align}
S(x)=A\exp(\zeta x),
\end{align}
for some constant $A$, and thus
\begin{align}
\label{O-asymptotic}
p(x)=\mathcal{O}\big(\exp(e^{\zeta x})\big),\ \text{as} \ x\rightarrow \infty.
\end{align}
\end{proof}
\begin{remark}
\label{Remark:6}
Due to Lemma 1, it is easy to verify that when $K\leq -\zeta r(\lambda/\zeta)$, the property of \eqref{O-asymptotic} does not hold in general. In fact, for sufficiently large negative $K$, solutions of Eq. \eqref{ODE} are decreasing power policies. However, we conjecture that all such power policies are suboptimal as they fail to control the overflow in the battery. A more detailed discussion will be presented in the following section.
\end{remark}

\section{Numerical Experiments}
We now study a multiple access communication system consisting of two nodes ($M=2$) with $\lambda_{1}=\lambda_{2}=\lambda=1$ and $\zeta_{1}=\zeta_{2}=\zeta=1$. Because of the symmetry of the MAC, the achievable power policies for this setting are obtained through Algorithm 2. However, to implement Algorithm 2 according to steps 1-6, one is obliged to search for optimized values of $p(0^{+})$ and $K$ at each iteration. To ease this process and in what follows, Algorithm 2 is modified in a way that once the values of $p(0^{+})$ and $K$ are initialized, the same values are used at each iteration step. \footnote{Although this approach is potentially suboptimal, it always yields achievable results, and in the case of the considered example here, the achievable results are close to the upper bound.}

 With this modification, Fig. \ref{Fig:1A} then shows the designed power policy as a function of the remaining charge in the battery with initial conditions $p(0^{+})=0.1$ and $K=0$ and initializing function $p^{(0)}(x_{k})=x_{k}+p(0^{+}), 0< x_{k}\leq L_{k}$. After $N=10$ iterations, the power policy has converged to a solution of \eqref{symm-MAC}. It can also be seen from Fig. \ref{Fig:1A} that as the remaining charge in the battery increases, the transmission power also increases rapidly. Supported by part (b) of Lemma \ref{Lemma:1}, we further conjecture that this increase is in fact doubly exponential in $x$. Indeed, when the occupied charge of the battery becomes large, the chance of overflow due to new energy arrivals increases as well. In this regard, an optimal power policy is one which consumes the battery charge fast enough such that the occurrence of overflow is traded-off against the potential suboptimality of employing a large instantaneous transmission power (see Section IV). On the other hand, for sufficiently large negative $K$, solutions of Eq. \eqref{symm-MAC} are non-increasing (see Remark \ref{Remark:6} for the point-to-point case) and they thus fail to manage battery overflow. The numerical results have further verified that for non-increasing power policies, the achieved sum-throughput is strictly less than for increasing ones. As a result, here we only consider increasing power policies.

Corresponding to the designed power policy in Figure \ref{Fig:1A}, Figure \ref{Fig:1B} shows the absolutely continuous part (density) of the probability measure in Eq. \eqref{dens-1}. In this case, consistent with our earlier observation for the power policy, the density function also falls off quickly. In terms of ergodicity, this is basically an assertion of the fact that the system spends little time with large stored charge in the battery.

\begin{table}
	\center
	\caption{Total average throughput for two identical nodes, using Shannon rate function, $r(x)={1\over 2}\log(1+x/N_{0})$, with $N_{0}=1$, equation constant $K = 0$, initializing function $p(x)=x+p(0^{+}), 0< x\leq L$, and for various storage capacity $L$ and initial values $p(0^{+})$.}
	\label{Table:1}
	\begin{tabular}{cccccc}
		\toprule
		\multicolumn{6}{c}{\hspace*{20mm}$\text{Initial Value}\ p(0^{+})$}\\
		\cline{3-6}
		$L$  & $R_{\text{upper}}$ & $0.001$ & $0.001$ & $0.1$ & $1$  \\
		\hline
		0.5    & 0.4187 & 0.3177 & 0.3152 & 0.3094  & 0.2797    \\
		1      & 0.5895 & 0.4217 & 0.4159 & 0.4069  & 0.3722    \\
		2      & 0.7243 & 0.4634 & 0.4575 & 0.4511  & 0.4075    \\
		3      & 0.7681 & 0.4652 & 0.4593 & 0.4510  & 0.4091    \\
		\bottomrule
	\end{tabular}
\end{table}
\begin{table}
	\center
	\caption{Total average throughput for two identical nodes, using Shannon rate function, $r(x)={1\over 2}\log(1+x/N_{0})$, with $N_{0}=1$, initial value $p(0^+) = 0.001$, initializing function $p(x)=x+p(0^+), 0< x\leq L$, and for various storage capacity $L$ and equation constant $K$. The upper bound for infinite storage battery ($L_{k}=\infty$) is given by $R_{\infty}={1\over 2}\log(1+2)=0.792$.}.
	\label{Table:2}
	\begin{tabular}{cccccc}
		\toprule
		\multicolumn{6}{c}{\hspace*{16mm}$\text{Equation Constant}\ K$} \\
		\cline{3-6}
		$L$ & $R_{\text{upper}}$ & $+0.5$ & $0$ & $-0.5$ & $\text{Optimum}$  \\
		\hline
		0.5   & 0.4187 & 0.3017 & 0.3177 & 0.3057 & 0.3262 (77.9\%)  \ [K=-0.15]   \\
		1     & 0.5895 & 0.3707 & 0.4217 & 0.4410 & 0.4612 (78.2\%)  \ [K=-0.37]   \\
		2     & 0.7243 & 0.3854 & 0.4634 & 0.5725 & 0.5951 (82.1\%)  \ [K=-0.63]   \\
		3     & 0.7681 & 0.3858 & 0.4652 & 0.5907 & 0.6654 (86.6\%)  \ [K=-0.67]   \\
		\bottomrule
	\end{tabular}
\end{table}

\begin{figure}[t!]
\begin{center}
 \subfigure[]{ \includegraphics[trim={.8cm .8cm .8cm .8cm}, width=1\linewidth]{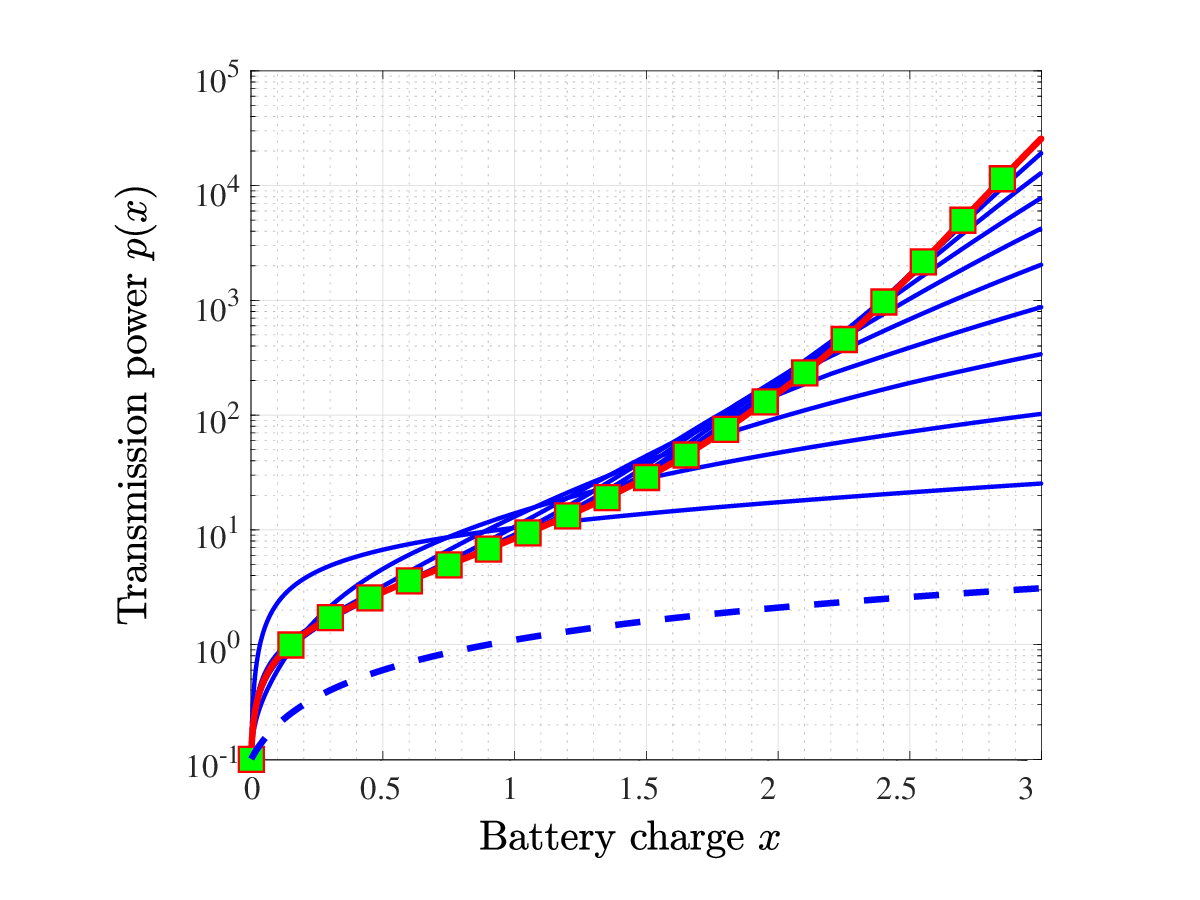}\label{Fig:1A}}\\ \vspace{5mm}
  ~ 
\subfigure[]{\includegraphics[trim={.8cm .8cm .8cm .8cm}, width=1\linewidth]{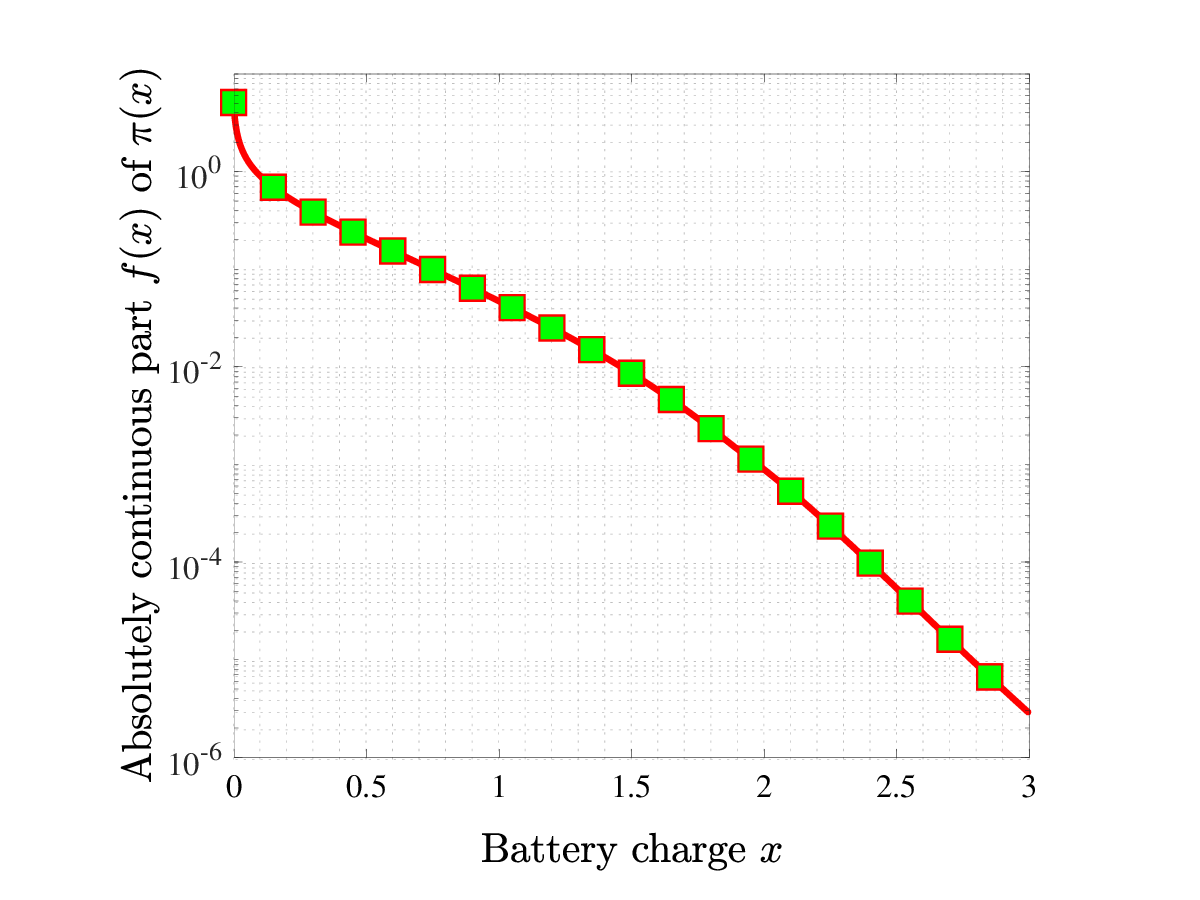}\label{Fig:1B}}
  ~ 
 \caption{\small{Battery capacity $L=3$, equation constant $K=0$, and $p(0^{+})=0.1$ for two nodes case. (a) The convergence of transmission power policy to an achievable policy (denoted by squares) after $N=10$ iterations with initializing function $p^{(0)}(x)=x+p(0^{+})$ (dashed lines) and iterates (solid lines), (b) Absolutely continuous part $f(x)$ of $\pi(x)$ for the converged solution. }}
 \end{center}
\end{figure}

 Using Algorithm \ref{Alg:2}, we have computed the achievable rates provided in Table \ref{Table:1} and Table \ref{Table:2}, where the termination criterion is taken to be
\begin{align}
\nonumber
\theta={r\big(\sum_{k=1}^{M}p^{N+1}(x_{k})\big)-r\big(\sum_{k=1}^{M}p^{N}(x_{k})\big)\over r\big(\sum_{k=1}^{M}p^{N}(x_{k})\big)} <1\%,
\end{align}
\textit{i.e.}, the iteration stops whenever the increase in rate is less than one percent. With this precision, Table \ref{Table:1} shows the sum throughput for several choices of $p(0^{+})$ and fixed $K=0$. The upper bound is also computed from \eqref{upper bound} and denoted by $R_{\text{upper}}$ in the table.

 Based on a comparison between the upper and lower limits on the rate function, it is immediate that the choice of $p(0^{+})=0.001$ results the best performance of the designed power policy. For the same choices of $p(0^{+})$ as in Table \ref{Table:1} and $K=0$, Figure \ref{Fig:2} shows the power policy solutions. Except for the case of $p(0^{+})=1$, all the power policy solutions adopt a small transmission power when the battery charge is small. Along the same lines, Table \ref{Table:2} shows the upper and lower limits on the average throughput for fixed $p(0^{+})=0.001$ and variable $K$. We have particularly provided the best value of $K$ up to precision $0.01$ as well as the corresponding achievable rates. The best achievable rate, as a percentage of the \textit{upper bound}, is also evaluated.

Finally, to show the robustness of the iterative algorithm to the initializing function, a different choice of $p^{(0)}_{k}(x_{k})$ is studied in Figure \ref{Fig:3}. Therein, we particularly have selected $p^{(0)}_{k}(x_{k})=p_{k}(0^{+}), 0< x_{k}\leq 3$ for purpose of initialization in Algorithm \ref{Alg:2} while the rest of the parameters are the same as in Figure \ref{Fig:1A}. Evidently, the power policy converges to an identical function as one depicted in Figure \ref{Fig:1A}. Similarly, the same convergence was observed when $p^{(0)}_{k}(x_{k})=p_{k}(0^{+})+\sqrt{x_{k}}, 0< x_{k}\leq 3$. In this respect, the proposed algorithm appears to be insensitive to the choice of the initial power policy.
\section{conclusion}
We have considered continuous-time power policies for a multiple access communication system where each node is capable of harvesting energy. First we modelled the battery as a compound Poisson dam, where the remaining charge in the battery modulates the transmission power. We then analysed this storage dam model in the ergodic case. In particular, we characterized an upper bound on the maximum sum-rate as a function of the energy arrivals distributions and the capacity of the batteries. For batteries with infinite capacity, we proved that any rate close to this bound is achievable by a set of constant power policies that result in stable battery behavior. For batteries with limited capacity, we showed that optimal power policies can be derived by solving a system of simultaneous partial integro-differential equations. To solve these equations, we developed an iterative algorithm based on the Gauss-Seidel approach. We next derived a fixed point algorithm for the symmetric MAC case where the multiple access nodes have identical energy harvesting statistics. Furthermore, the convergence of the utility function that results from the proposed algorithms was established. Numerical results show that for $L=3$, the achievable scheme provides throughput up to $86.6\%$ of the upper bound.

Potential future work includes extending the study to the case where each transmitter has a data buffer. This could model scenarios in which a sensor monitors a physical quantity (e.g. temperature), and then stores the data in a buffer for eventual transmission once enough energy has been harvested.

\appendices
\section{}
In the following, we show that for every power policy with memory, $p^{*}_{k}(X_{k}(u);u\leq t)$, there exist a memoryless counterpart $p_{k}(X_{k}(t))$ that attains the same or better sum-throughput performance. In particular, let $(\Omega,\mathcal{F},\prob)$ be a complete probability space and $X^{*}_{k}(t;\omega), -\infty<t<\infty$ be a stationary and ergodic stochastic process defined on this probability space and whose evolution for $t\geq 0$ is given by\footnote{For clarity, we make the dependence on $\omega\in \Omega$ explicit.}
\begin{align}
\nonumber
&X^{*}_{k}(t;\omega)=X^{*}_{k}(0;\omega)+E_{k}^{\mathrm{In}}\big((0,t];\omega\big)\\ \label{st-ext} &\hspace{30mm}-\int_{0}^{t}p^{*}_{k}\big(X^{*}_{k}(u;\omega);u\leq s\big)\ ds.
\end{align}
In conjunction with the process $X^{*}_{k}(t;\omega)$, we then define the following empirical CDFs,
\begin{align}
\label{def-empricalCDF}
\hspace{-.5mm}&\widetilde{F}^{*}_{k}(\rho_{k},x_{k};\omega)\triangleq \\ \nonumber &\lim_{T\rightarrow \infty} {1\over T}\hspace{-.5mm}\int_{0}^{T}\hspace{-2.2mm}\mathbf{1}\big(p_{k}(X^{*}_{k}(u;\omega);u\leq s)\leq \rho_{k}\big)\mathbf{1}(X^{*}_{k}(s;\omega)\leq x_{k}) ds\\
&\widetilde{\pi}^{*}_{k}(x_{k};\omega)\triangleq \lim_{T\rightarrow \infty} {1\over T}\int_{0}^{T}\mathbf{1}(X^{*}_{k}(s;\omega)\leq x_{k})ds \\
&\hspace{14mm}=\lim_{\rho_{k}\rightarrow \infty}\widetilde{F}^{*}_{k}(\rho_{k},x_{k};\omega).
\end{align}
Now since $X^{*}_{k}(t)$ is ergodic, $\widetilde{F}^{*}_{k}(\rho_{k},x_{k};\omega)$ and $\widetilde{\pi}^{*}_{k}(x_{k};\omega)$ are well defined, and constant $\prob-$almost surely on $\Omega$, \textit{i.e.},
\begin{subequations}
	\label{Eq:0111x1}
\begin{align}
\label{Eq:0111x}
\widetilde{F}^{*}_{k}(\rho_{k},x_{k};\omega)&\stackrel{\text{a.s.}}{=}F_{k}^{*}(\rho_{k},x_{k}), \\
\label{Eq:0222x}
\widetilde{\pi}^{*}_{k}(x_{k};\omega)&\stackrel{\text{a.s.}}{=}\pi_{k}^{*}(x_{k}),
\end{align}
\end{subequations}
For the functions $F_{k}^{*}(\rho_{k},x_{k})$, $\pi_{k}^{*}(x_{k})$ in Eqs. \eqref{Eq:0111x1} we define the conditional CDF $F_{k}^{*}(\rho_{k}\vert x_{k})$ by
\begin{align}
F^{*}_{k}(\rho_{k},x_{k})=\int_{0}^{x_{k}}F_{k}^{*}(\rho_{k}\vert s)\pi_{k}^{*}(ds).
\end{align}
Also, we define the memoryless power policy $p_{k}(x_{k})$ as follows
\begin{align}
\label{pstar}
p_{k}(x_{k})\triangleq \int_{0}^{\infty}\rho_{k}F^{*}_{k}(d\rho_{k}\vert x_{k}),
\end{align}
and a corresponding storage process $X_{k}(t)$ governed by
\begin{align}
\label{ss-1}
X_{k}(t)=X_{k}(0)+E_{k}^{\mathrm{In}}(0,t]-\int_{0}^{t}p_{k}\big(X_{k}(s)\big)\d s.
\end{align}
Also, we denote the stationary measure of $X_{k}(t)$ by $\pi_{k}(x_{k})$. Our objective now is to prove that the throughput using the storage process with memory $p^{*}_{k}(X^{*}_{k}(u;w),u\leq t)$ is no better than that of it's memoryless counterpart $p_{k}(X_{k}(t))$, \textit{i.e.},
\small\begin{align}
\label{tt1}
\widehat{R}(\{p_{k}^{*}(X^{*}_{k}(s;w);s\leq t)\}_{k=1}^{M}) {\leq} \widehat{R}\big(\{p_{k}(X_{k}(t))\}_{k=1}^{M}\big),
\end{align}\normalsize
almost surely. To show this result, we begin with the definition of the long term average throughput for the storage process in Eq. \eqref{st-ext}, \textit{i.e.},
\begin{figure}[t]
	\begin{center}
		\includegraphics[trim={.8cm .8cm .8cm .8cm}, width=1\linewidth]{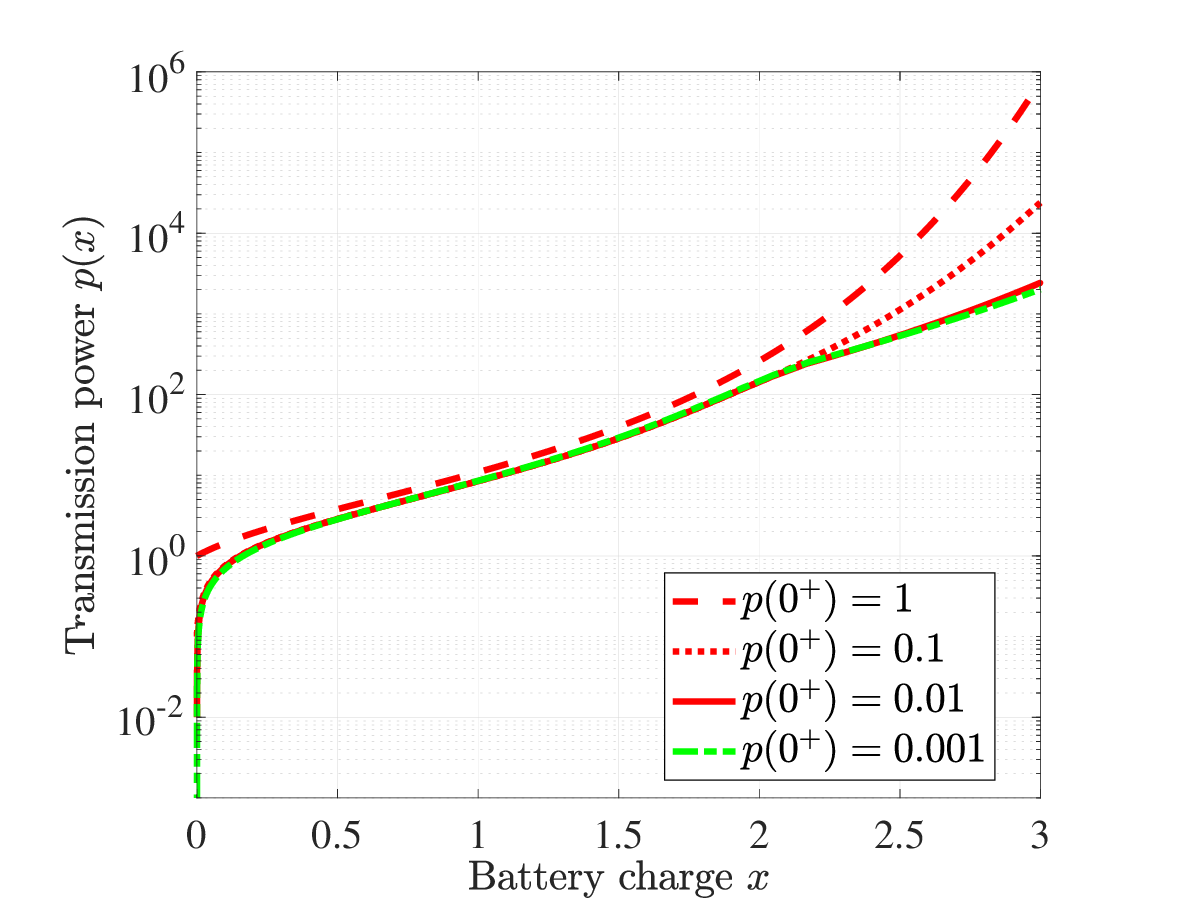}
		\vspace{2mm}
		\caption{Transmission power policies $p(x)$ with different initial values ($L=3,M=2,K=0$).}
		\label{Fig:2}
	\end{center}
\end{figure}
\begin{figure}[t]
	\begin{center}
		\centering
		\includegraphics[trim={.8cm .8cm .8cm .8cm}, width=1\linewidth]{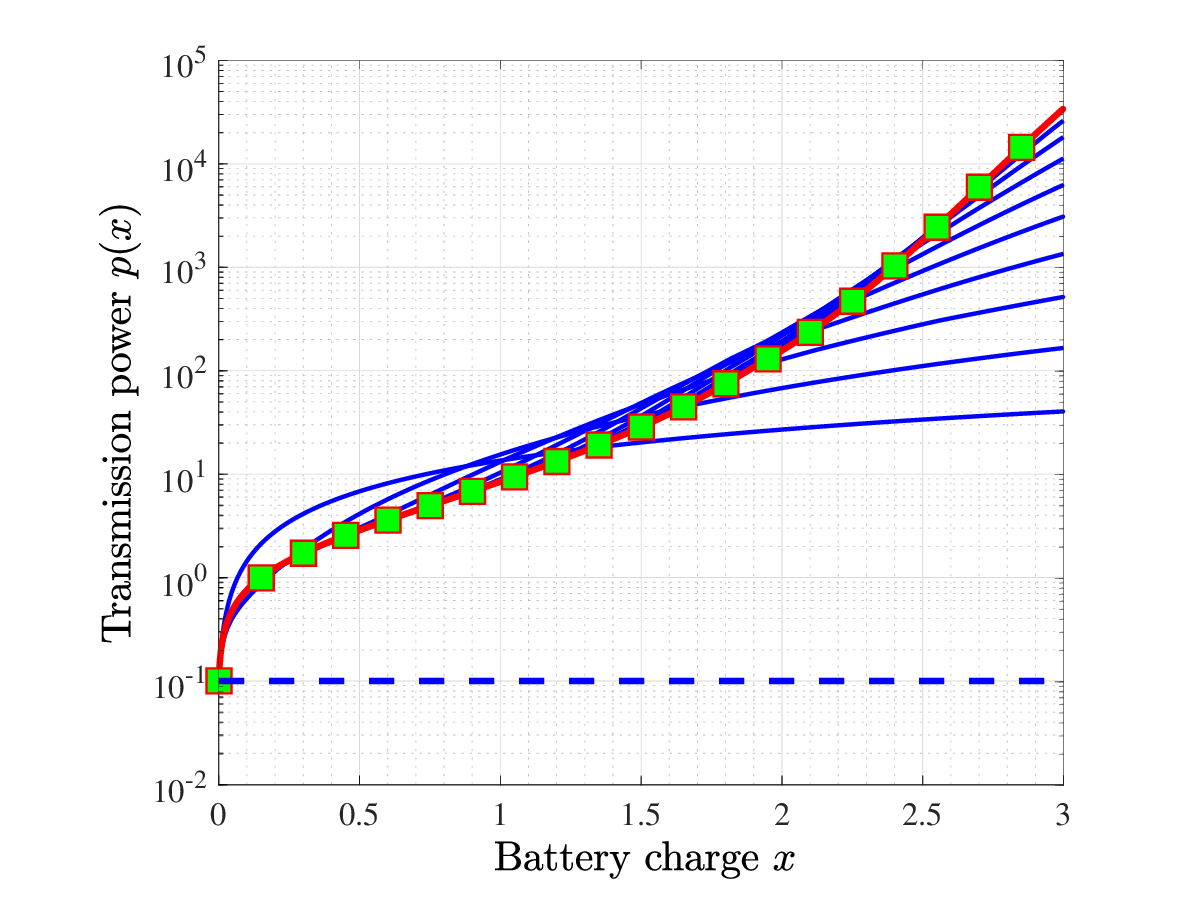}
		\vspace{.3mm}
		\caption{Robustness to the initializing function, using a constant initializing function (dashed line) $p^{(0)}(x_{k})=p(0^{+}), 0\leq x_{k}\leq L$ ($L=3,M=2,K=0$ and $N=10$).}
		\label{Fig:3}
	\end{center}
\end{figure}

\begin{align}
\nonumber
&\widehat{R}\big( \{p^{*}_{k}(X^{*}_{k}(s;w);s\leq t)\}_{k=1}^{M}\big)\\ &\triangleq \lim_{T\rightarrow \infty}{1\over T}\int_{0}^{T}r\Big(\sum_{k=1}^{M}p^{*}_{k}(X^{*}_{k}(s;w);s\leq t)\Big)\d s\\
&\stackrel{\rm{(a)}}{=} \int_{\mathcal{A}}\int_{{\mathcal{B}}}r(\sum_{k=1}^{M}\rho_{k})\prod_{k=1}^{M}\widetilde{F}^{*}_{k}(d\rho_{k},\d x_{k};w)\\
&\stackrel{\text{a.s.}}{=}\int_{\mathcal{A}}\int_{{\mathcal{B}}}r(\sum_{k=1}^{M}\rho_{k})\prod_{k=1}^{M}F^{*}_{k}(d\rho_{k},\d x_{k}),  \\
\label{up-shannon}
&=\int_{\mathcal{A}}\int_{{\mathcal{B}}}r(\sum_{k=1}^{M}\rho_{k})\prod_{k=1}^{M}F^{*}_{k}(d\rho_{k}\vert x_{k})\pi^{*}_{k}(\d x_{k}),
\end{align}
where \rm{(a)} follows from the definition of $\widetilde{F}^{*}_{k}(\d\rho_{k},\d x_{k};w)$ in Eq. \eqref{def-empricalCDF}, and $\mathcal{B}\triangleq [0,\infty)\times [0,\infty)\times \cdots \times [0,\infty)$ is the domain of integration on $\{\rho_{k}\}_{k=1}^{M}$. From concavity of the rate function, we then upper bound \eqref{up-shannon} as follows
\begin{align}
\nonumber
&\widehat{R}\big(\{p^{*}_{k}(X^{*}_{k}(s;w);s\leq t)\}_{k=1}^{M}\big)\\ \nonumber &\stackrel{\text{a.s.}}{=}\int_{\mathcal{A}}\int_{{\mathcal{B}}}r\Big(\sum_{k=1}^{M}\rho_{k}\Big)\prod_{\ell=1}^{M}F^{*}_{\ell}(d\rho_{\ell}\vert x_{\ell})\prod_{k=1}^{M}\pi^{*}_{k}(\d x_{k}),\\ \nonumber
&\leq \int_{\mathcal{A}}r(\sum_{k=1}^{M}\int_{B}\rho_{k}\prod_{\ell=1}^{M}F^{*}_{\ell}(d\rho_{\ell}\vert x_{\ell}))\prod_{k=1}^{M}\pi^{*}_{k}(\d x_{k})
\\ \label{1010}&= \int_{\mathcal{A}}r\big(\sum_{k=1}^{M}\int_{0}^{\infty}\rho_{k}F^{*}_{k}(d\rho_{k}\vert x_{k})\big)\prod_{k=1}^{M}\pi^{*}_{k}(\d x_{k})
\\ \label{343}&\triangleq \int_{\mathcal{A}}r\big(\sum_{k=1}^{M}p_{k}(x_{k})\big)\prod_{k=1}^{M}\pi^{*}_{k}(\d x_{k}),
\end{align}
where the last step follows from the definition of $p_{k}\big(x_{k}\big)$ in Eq. \eqref{pstar}. The remaining task is now to show that $\pi^{*}_{k}(x_{k})=\pi_{k}(x_{k}),\forall k\in [M]$. For this purpose, we define some notation in conjunction with an arbitrary, stationary,  \textit{c\`{a}dl\`{a}g}\footnote{Right continuous with left hand limit. Note that both the storage processes in Eqs. \eqref{st-ext} and \eqref{ss-1} are c\`{a}dl\`{a}g.} process $Y(t)$ whose jumps (positive or negative) occur at time instants $T^{0},T^{1},\cdots$. In particular, the right hand derivative of $Y(t)$ is defined by
\begin{align}
Y^{+}(t)\triangleq \lim_{\varepsilon \downarrow 0}{Y(t+\varepsilon)-Y(t)\over \varepsilon}.
\end{align}
In addition, define
\begin{align}
Y(t^{-})\triangleq \lim_{\varepsilon \downarrow 0} Y(t-\varepsilon).
\end{align}
\begin{theorem}(\textit{Rate Conservation Law})
\label{Theorem:1}
Let $Y(t)$ be an ergodic, stationary, \textit{c\`{a}dl\`{a}g} process. Then,
\begin{align}
\label{general_level_crossing}
&f(y)\expect\big[Y^{+}(t)\vert Y(t)=y\big]= \\ \nonumber &\lambda^{0}\expect^{0}\big[\mathbf{1}_{\{Y(T^{0,-})>y\}}\mathbf{1}_{\{Y(T^{0})<y\}}-\mathbf{1}_{\{Y(T^{0,-})<y\}}\mathbf{1}_{\{Y(T^{0})>y\}}\big],
\end{align}
where $f(y)$ is the probability density at $y$, and $\expect^{0}$ denotes the expectation with respect to the Palm probability distribution corresponding to the point process (with assumed intensity $\lambda^{0}$) for the jumps. $\hfill\square$
\end{theorem}
\begin{proof}
The proof can be found in \cite[p. 36]{Mazumdar1}.
\end{proof}
\begin{remark}
The term $\mathbf{1}_{\{Y(T^{0,-})>y\}}\mathbf{1}_{\{Y(T^{0})<y\}}$ in the right hand side of Theorem \ref{Theorem:1} corresponds to negative jumps in the sample path while $\mathbf{1}_{\{Y(T^{0,-})<y\}}\mathbf{1}_{\{Y(T^{0})>y\}}$ corresponds to positive jumps.
\end{remark}
\begin{remark}
\label{Remark:8}
As the memoryless storage process in Eq. \eqref{ss-1} only contains positive jumps,
\begin{align}
\nonumber
\expect^{0}\left[\mathbf{1}_{\{X_{k}(T_{k}^{0,-})>x_{k}\}}\mathbf{1}_{\{X_{k}(T_{k}^{0})<x_{k}\}}\right]=0,
\end{align}
where as defined in Section \ref{section:Preliminaries}-B,  $T^{0}_{k},T^{1}_{k},\cdots$ denote the energy arrival times for the $k^{th}$ node. In this special case, we then have
\footnotesize\begin{align}
\nonumber
\hspace{-8mm}f_{k}(x_{k})\expect[X_{k}^{+}(t)\vert X_{k}(t)=x_{k}]
&=f_{k}(x_{k})\expect[-p_{k}(X_{k}(t))\vert X_{k}(t)=x_{k}]\\ \label{1368}
&=-f_{k}(x_{k})p_{k}(x_{k}).
\end{align}\normalsize
For the right hand side of Theorem \ref{Theorem:1} we obtain
\begin{align}
\nonumber
&\lambda^{0}\expect^{0}[-\mathbf{1}_{\{X_{k}(T^{0,-})<x_{k}\}}\mathbf{1}_{\{X_{k}(T^{0})>x_{k}\}}]\\ 
&\stackrel{\rm{(a)}}{=}\lambda_{k} \expect[-\mathbf{1}_{\{X_{k}(T^{0,-})<x_{k}\}}\mathbf{1}_{\{X_{k}(T^{0})>x_{k}\}}]\\
&=-\lambda_{k}\int_{0}^{x_{k}}\big(1-B_{k}(x_{k}-v_{k})\big)\pi_{k}(\d v_{k})\\ \nonumber
&=-\lambda_{k}\Big[\big(1-B_{k}(x_{k})\big)\pi_{k}^{0}\\  \label{1369} &\hspace{13mm}+\int_{0^{+}}^{x_{k}}\big(1-B_{k}(x_{k}-v_{k})\big)f(v_{k}) \d v_{k}\Big],
\end{align}
where \rm{(a)} follows from the notion of Poisson Arrivals See Time Averages (PASTA) \cite[Prop. 1.23]{Mazumdar1} for Poisson energy arrival process  with intensity $\lambda^{0}=\lambda_{k}$.
Equating \eqref{1368} and \eqref{1369} according to Theorem \ref{Theorem:1}, we obtain
\begin{align}
\nonumber
\hspace{-3mm}f_{k}(x_{k})p_{k}(x_{k})&=\lambda_{k}\Big[\big(1-B_{k}(x_{k})\big)\pi_{k}^{0}\\ \label{hh-1} &+\int_{0^{+}}^{x_{k}}\big(1-B_{k}(x_{k}-v_{k})\big)f_{k}(v_{k}) \d v_{k}\Big],
\end{align}
which is the equilibrium condition in Eq. \eqref{reflecct} with the density $f_{k}(x_{k})$ and the atom $\pi_{k}^{0}$.
\end{remark}

Returning to the storage process with memory in Eq. \eqref{st-ext}, now it is also easy to see that
\begin{align}
\expect\big[(X_{k}^{*}(t))^{+}\vert X_{k}^{*}(t)=x_{k}\big]&= \int_{0}^{\infty}\rho_{k} F_{k}(\d\rho_{k}\vert x_{k})\\
&\triangleq p_{k}(x_{k}),
\end{align}
which is simply the average rate of down crossing at level $x_{k}$ corresponding to the stationary distribution of $X^{*}_{k}(t)$ in Eq. \eqref{st-ext}. Using the argument in Remark \ref{Remark:8} for the process in Eq. \eqref{st-ext} results
\begin{align}
\nonumber
&f^{*}_{k}(x_{k})\expect\big[(X_{k}^{*}(t))^{+}\vert X_{k}^{*}(t)=x_{k}\big]=f^{*}_{k}(x_{k})p_{k}(x_{k})\\ \nonumber
&=\lambda_{k}\bigg[\big(1-B_{k}(x_{k})\big)\pi^{*,0}_{k}\\ \label{hh-2} &\hspace{10mm}+\int_{0^{+}}^{x_{k}}\big(1-B_{k}(x_{k}-v_{k})\big)f_{k}^{*}(v_{k}) \d  v_{k}\bigg],
\end{align}
where $\pi^{*,0}_{k}$ and $f_{k}^{*}(x_{k})$ are the atom and the continuous part (density) of the probability measure $\pi_{k}^{*}(x_{k})$. Since from Theorem 1 the probability measure that solves \eqref{hh-1} and \eqref{hh-2} is unique,
\begin{align}
\pi^{*}_{k}(x_{k})=\pi_{k}(x_{k}), \ \ \ \forall x_{k}.
\end{align}
Concluding from \eqref{343}, we thus showed that
\begin{align}
\nonumber
\widehat{R}\big( \{p^{*}_{k}(X^{*}_{k}(s;w);s\leq t)\}_{k=1}^{M}\big)
&\stackrel{\text{\text{a.s.}}}{\leq} \int_{\mathcal{A}}r\big(\sum_{k=1}^{M}p_{k}(x_{k})\big)\pi^{*}_{k}(\d x_{k}), \\ \nonumber
&= \int_{\mathcal{A}}r\big(\sum_{k=1}^{M}p_{k}(x_{k})\big)\pi_{k}(\d x_{k})\\ \nonumber
&\stackrel{\text{a.s.}}{=}\widehat{R}\big(\{p_{k}(X_{k}(t))\}_{k=1}^{M}\big), 
\end{align}
\begin{remark}
We note that in the ergodic regime, the upcrossing rate as well as the drift component of the storage process in the finite battery case also obey the law stated in Theorem \ref{Theorem:1}. Thus, we again obtain Eq. \eqref{hh-2} as battery overflow does not change the upward and downward rates. Therefore, a similar proof can be used to show $X_{k}(t)$ is a sufficient statistic for optimal power policies in the storage model with a finite battery capacity in Eq. \eqref{storage model1}.
\end{remark}

\section{}
\begin{proof}
Since for all $ k\neq j$,
\begin{align}
\nonumber
(\pi_{k}^{0,\alpha},g_{k}^{\alpha}(x_{k}))=(\pi_{k}^{0,1},{g}_{k}^{1}(x_{k}))=({\pi}_{k}^{0,2},{g}_{k}^{2}(x_{k})),
\end{align}
we have that
\begin{align}
\nonumber
\pi_{k}^{\alpha}(\d x_{k})=\pi_{k}^{1}(\d x_{k})={\pi}^{2}_{k}(\d x_{k}).
\end{align}
Then,
\small\begin{align}
\nonumber
&\widehat{R}_{j}^{\alpha}=\int_{\mathcal{A}}r\bigg(\lambda_{j}{ {G^{\alpha}_{j}(x_{j})}\over {g^{\alpha}_{j}(x_{j})}}+\hspace{-3mm}\sum_{k\in [M]-j}\hspace{-3mm}\lambda_{k}{G_{k}(x_{k})\over g_{k}(x_{k})}\bigg) \pi_{j}^{\alpha}(\d x_{j})\hspace{-3mm}\prod_{k\in [M]-j}\hspace{-2.5mm}\pi_{k}(\d x_{k}) \\  &=\expect_{j}\bigg[\int_{0}^{L_{j}}r\bigg(\lambda_{j}{ {G^{\alpha}_{j}(x_{j})}\over {g^{\alpha}_{j}(x_{j})}}+\sum_{k\in [M]-j}\hspace{-2mm}\lambda_{k}{G_{k}(x_{k})\over g_{k}(x_{k})}\bigg)\pi_{j}^{\alpha}(\d x_{j})\bigg]\\ \nonumber
&=\expect_{j}\bigg[\int_{0}^{L_{j}}r\bigg(\lambda_{j}{ {G^{\alpha}_{j}(x_{j})}\over {g^{\alpha}_{j}(x_{j})}}+\sum_{k\in [M]-j}\lambda_{k}{G_{k}(x_{k})\over g_{k}(x_{k})}\bigg)\\ &\hspace{4mm} \times [\pi_{j}^{0,\alpha}\delta(x_{j})+e^{-\zeta_{j}x_{j}}g^{\alpha}_{j}(x_{j})]\d x_{j}\bigg]\\ \nonumber
&=\expect_{j}\bigg[\hspace{-1mm}\int_{0}^{L_{j}}\hspace{-3mm}r\bigg(\lambda_{j}{ {G^{\alpha}_{j}(x_{j})}\over {g^{\alpha}_{j}(x_{j})}}+\hspace{-3mm}\sum_{k\in [M]-j}\hspace{-3mm}\lambda_{k}{G_{k}(x_{k})\over g_{k}(x_{k})}\bigg)e^{-\zeta_{j}x_{j}}g^{\alpha}_{j}(x_{j})\d x_{j}\bigg]
\\ \label{a number} &\hspace{4mm}+\pi_{j}^{0,\alpha}\expect_{j}\bigg[r\bigg(\sum_{k\in [M]-j}\hspace{-2mm}\lambda_{k}{G_{k}(x_{k})\over g_{k}(x_{k})}\bigg)\bigg].
\end{align}\normalsize
For the term inside the first expectation in Eq. \eqref{a number}, we proceed as \eqref{Stan:01}-\eqref{Stan:06} on the next page,
\begin{figure*}[!t]
\normalsize
\begin{align}
\nonumber
&\int_{0}^{L_{j}}r\bigg(\lambda_{j}{ {G^{\alpha}_{j}(x_{j})}\over {g^{\alpha}_{j}(x_{j})}}+\sum_{k\in [M]-j}\lambda_{k}{G_{k}(x_{k})\over g_{k}(x_{k})}\bigg) e^{-\zeta_{j}x_{j}}g^{\alpha}_{j}(x_{j})\d x_{j}
\\ \label{Stan:01} & =\int_{0}^{L_{j}}r\bigg(\lambda_{j}{\alpha {G^{1}_{j}(x_{j})+\bar{\alpha}{G}^{2}_{j}(x_{j})}\over {\alpha g^{1}_{j}(x_{j})+\bar{\alpha}{g}^{2}_{j}(x_{j})}}+\sum_{k\in [M]-j}{\lambda_{k} G_{k}(x_{k})\over g_{k}(x_{k})}\bigg) \big({\alpha e^{-\zeta_{j}x_{j}}g_{j}^{1}(x_{j})+\bar{\alpha} e^{-\zeta_{j}x_{j}}{g}^{2}_{j}(x_{j})}\big)\d x_{j}\\ 
&=\int_{0}^{L_{j}}r\bigg(\lambda_{j}{\alpha e^{-\zeta_{j}x_{j}}{G^{1}_{j}(x_{j})+\bar{\alpha}e^{-\zeta_{j}x_{j}}{G}^{2}_{j}(x_{j})}\over {\alpha e^{-\zeta_{j}x_{j}}g^{1}_{j}(x_{j})+\bar{\alpha}e^{-\zeta_{j}x_{j}}{g}^{2}_{j}(x_{j})}}+\sum_{k\in [M]-j}{\lambda_{k} G_{k}(x_{k})\over g_{k}(x_{k})}\bigg) \big({\alpha e^{-\zeta_{j}x_{j}}g_{j}^{1}(x_{j})+\bar{\alpha} e^{-\zeta_{j}x_{j}}{g}^{2}_{j}(x_{j})}\big)\ d x_{j}\\ \nonumber 
& {\geq} \int_{0}^{L_{j}} r\bigg(\lambda_{j}{{\alpha e^{-\zeta_{j}x_{j}}G_{j}^{1}(x_{j})}\over {\alpha e^{-\zeta_{j}x_{j}}g_{j}^{1}(x_{j})}}+\sum_{k\in [M]-j} {\lambda_{k} G_{k}(x_{k})\over g_{k}(x_{k})}\bigg)\alpha e^{-\zeta_{j}x_{j}}g_{j}^{1}(x_{j})\ d x_{j}\\ & \hspace{4mm}\label{eq:harvardtarokh}+ \int_{0}^{L_{j}} r\bigg(\lambda_{j}{{\bar{\alpha}e^{-\zeta_{j}x_{j}}{G}^{2}_{j}(x_{j})}\over {\bar{\alpha}e^{-\zeta_{j}x_{j}}{g}^{2}_{j}(x_{j})}}+\sum_{k\in [M]-j}{\lambda_{k} {G}_{k}(x_{k})\over {g}_{k}(x_{k})}\bigg) \bar{\alpha} e^{-\zeta_{j}x_{j}}{g}^{2}_{j}(x_{j}) d x_{j}
\\  \nonumber
& = \alpha \int_{0}^{L_{j}} r\bigg(\lambda_{j}{{G_{j}^{1}(x_{j})}\over {g_{j}^{1}(x_{j})}}+\sum_{k\in [M]-j} {\lambda_{k} G_{k}(x_{k})\over g_{k}(x_{k})}\bigg) e^{-\zeta_{j}x_{j}}g_{j}^{1}(x_{j})\ d x_{j}\\ \label{Stan:06} & \hspace{4mm}+ \bar{\alpha}\int_{0}^{L_{j}} r\bigg(\lambda_{j}{{{{G}^{2}_{j}(x_{j})}\over {g}^{2}_{j}(x_{j})}}+\sum_{k\in [M]-j}{\lambda_{k} {G}_{k}(x_{k})\over {g}_{k}(x_{k})}\bigg)  e^{-\zeta_{j}x_{j}}{g}^{2}_{j}(x_{j}) d x_{j},
\end{align}
\hrule
\begin{align}
\nonumber
\hspace{-27mm}\expect_{j}\bigg[\int_{0}^{L_{j}}r\bigg(\lambda_{j}{ {G^{\alpha}_{j}(x_{j})}\over {g^{\alpha}_{j}(x_{j})}}&+\sum_{k\in [M]-j}\lambda_{k}{G_{k}(x_{k})\over g_{k}(x_{k})}\bigg)e^{-\zeta_{j}x_{j}}g^{\alpha}_{j}(x_{j})\d x_{j}\bigg]\\
\nonumber &\geq \alpha \expect_{j}\bigg[\int_{0}^{L_{j}}r\bigg(\lambda_{j}{ {G_{j}^{1}(x_{j})}\over {g_{j}^{1}(x_{j})}}+\sum_{k\in [M]-j}\lambda_{k}{G_{k}(x_{k})\over g_{k}(x_{k})}\bigg)e^{-\zeta_{j}x_{j}}g_{j}^{1}(x_{j})\d x_{j}\bigg]\\ \label{119} &\hspace{4mm}+\bar{\alpha}\expect_{j}\bigg[\int_{0}^{L_{j}}r\bigg(\lambda_{j}{ {{G}^{2}_{j}(x_{j})}\over {{g}^{2}_{j}(x_{j})}}+\sum_{k\in [M]-j}\lambda_{k}{G_{k}(x_{k})\over g_{k}(x_{k})}\bigg) e^{-\zeta_{j}x_{j}}{g}^{2}_{j}(x_{j})\d x_{j}\bigg].
\end{align}
\hrule
\end{figure*}
where \eqref{eq:harvardtarokh} can be verified via the lemma given in Appendix C and choosing
\begin{align}
\nonumber
&a_{1}=\alpha e^{-\zeta_{j}x_{j}}g_{j}^{1}(x_{j}),\ \ \ a_{2}=\bar{\alpha} e^{-\zeta_{j}x_{j}}{g}^{2}_{j}(x_{j}),\\ \nonumber
&b_{1}=\alpha e^{-\zeta_{j}x_{j}}G_{j}^{1}(x_{j}),\ \ \ b_{2}=\bar{\alpha} e^{-\zeta_{j}x_{j}}G^{2}_{j}(x_{j}),
\end{align}
and
\begin{align}
\nonumber
&\gamma = \lambda_{j},\ \ \ \ \ \beta = \sum_{{k\in [M]-j}}{\lambda_{k} G_{k}(x_{k})\over g_{k}(x_{k})}.
\end{align}
Therefore, for the first term of \eqref{a number} we obtain Eq. \eqref{119}. 

Splitting the second term of \eqref{a number} as
\small\begin{align}
\nonumber 
\pi_{j}^{0,\alpha}\expect_{j}\bigg[r\bigg(\sum_{k\in [M]-j}\hspace{-3mm}\lambda_{k}{G_{k}(x_{k})\over g_{k}(x_{k})}\bigg)\bigg]&=\alpha \pi_{j}^{0,1} \expect_{j}\bigg[r\bigg(\sum_{k\in [M]-j}\hspace{-3mm}\lambda_{k}{G_{k}(x_{k})\over g_{k}(x_{k})}\bigg)\bigg]\\ \label{120} &+\bar{\alpha}{\pi}_{j}^{0,2}\expect_{j}\bigg[r\bigg(\sum_{k\in [M]-j}\hspace{-3mm}\lambda_{k}{G_{k}(x_{k})\over g_{k}(x_{k})}\bigg)\bigg],
\end{align}\normalsize
and combining  \eqref{119} and \eqref{120} we derive
\begin{align}
\widehat{R}_{j}^{\alpha}\geq \alpha \widehat{R}^{1}_{j}+\bar{\alpha}\widehat{R}^{2}_{j}.
\end{align}
\end{proof}
\section{}
\begin{lemma}
Let $\gamma,\beta>0$, $a_{k}>0$ and $b_{k}>0$ be given. Then
\begin{align}
\sum_{k}a_{k}r\left(\gamma {b_{k}\over a_{k}} +\beta\right)\leq  ar\left(\gamma {b\over a} +\beta\right),
\end{align}
where $a=\sum_{k}a_{k}$ and $b=\sum_{k}b_{k}$.\hfill$\square$
\end{lemma}
\begin{proof}
we define the function $V(x)\triangleq xr\big({(\gamma/ x)}+\beta \big)$ which is known to be concave for all $x> 0$ since
\begin{align}
\nonumber
V''(x)={\gamma^{2}\over x^{3}}r''\left({\gamma\over x}+\beta\right)<0,
\end{align}
where the concavity property of the rate function has been used. We then proceed as
\begin{align}
\nonumber
\sum_{k}a_{k}r\left(\gamma {b_{k}\over a_{k}} +\beta\right)&=\sum_{k} b_{k}(a_{k}/b_{k})r\left(\gamma {b_{k}\over a_{k}} +\beta\right)\\ \nonumber
&=\sum_{k} b_{k}V(a_{k}/b_{k})\\ \nonumber
&=b\sum_{k} (b_{k}/b)V(a_{k}/b_{k}).
\end{align}
Furthermore, from concavity of $V(x)$,
\begin{align}
\nonumber
b\sum_{k} (b_{k}/b)V(a_{k}/b_{k})&\leq b V(\sum_{k}b_{k}/b \times a_{k}/b_{k})\\ \nonumber
&=bV(a/b)\\ \nonumber
&=  ar\left(\gamma {b\over a} +\beta\right).
\end{align}
Hence,
\begin{align}
\nonumber
\sum_{k}a_{k}r\left(\gamma {b_{k}\over a_{k}} +\beta\right) \leq ar\left(\gamma {b\over a} +\beta\right).
\end{align}
\end{proof}
\bibliographystyle{IEEEtran}
\bibliography{mybib1}

\end{document}